\documentclass[11pt,a4paper]{article}
\usepackage{etex}
\usepackage[utf8]{inputenc}
\usepackage[TS1,T1]{fontenc}
\usepackage{textcomp}
\usepackage[a4paper,tmargin=3cm,bmargin=3cm,rmargin=2.2cm,lmargin=2.2cm]{geometry}
\usepackage{lmodern}
\usepackage{mathrsfs}

\usepackage{amsfonts,amssymb,amsmath,mathtools,cite,slashed,cases}
\usepackage[amsthm,thmmarks,amsmath]{ntheorem}
\usepackage{slashed}
\usepackage{graphicx}

\usepackage{tikz} % for diagram
\usepackage[section]{placeins} % to force the diagram to be where it has to be
\usetikzlibrary{arrows}
\usepackage[toc,page]{appendix}

\usepackage{color}
\usepackage[french,english]{babel}
\usepackage[pdftex]{hyperref}

\evensidemargin=\oddsidemargin
\DeclareMathOperator{\Ker}{Ker}             %% Kernel
\DeclareMathOperator{\Tr}{Tr}                   %% trace of operator
\newcommand{\vc}{\vcentcolon =}            %%  :=  in definition
\newcommand{\cv}{=\vcentcolon }            %%  =:  in definition

\newtheorem{assumption}{Assumption}[section]
\newtheorem{theorem}[assumption]{Theorem}
\newtheorem{corollary}[assumption]{Corollary}

\newtheorem{lemma}[assumption]{Lemma}
\newtheorem{definition}[assumption]{Definition}
\newtheorem{remark}[assumption]{Remark}
\newtheorem{prop}[assumption]{Proposition}
\newtheorem{question}[assumption]{Question}

\newcommand{\abs}[1]{\vert#1\vert}        %% |x|

\newcommand{\A}{\mathcal{A}}                %% an algebra
\newcommand{\bbbone}{{\text{\usefont{U}{dsss}{m}{n}\char49}}}   %% Unit operator
\newcommand{\berg}{A^2(\mathbb{B}^n)}  %% Bergman space on B^n
\newcommand{\bergon}[1]{A^2(#1)}      %% Bergman space on B^n
\newcommand{\bergonw}[2]{A^2_{#1}(#2)} %% weighted Bergman space A_X^2(Y)
\newcommand{\C}{\mathbb{C}}                %% complex numbers
\newcommand{\commut}[2]{[ #1,\,#2 ] }   % [.,.]
\newcommand{\corref}[1]{Corollary \ref{#1}} %% Corollary X
\newcommand{\cstFock}[0]{t}                   %% parameter of Fock representation
\newcommand{\cstSch}[0]{\cstFock '}   % parameter of Schrödinger representation
\newcommand{\defeql}[0]{\mathrel{\vcenter{\baselineskip0.5ex \lineskiplimit0pt
                     \hbox{\scriptsize.}\hbox{\scriptsize.}}}%
                     =}  % :=
\newcommand{\dirac}{\slashed D}             %% Dirac slashed
\newcommand{\Coo}{C^\infty}                   %% space of smooth functions
\newcommand{\DD}{\mathcal{D}}             %% Dirac operator
\newcommand{\envalg}[1]{\mathscr{U}(#1)} %% envelopping algebra U(x)
\newcommand{\evalat}[1]{|_{#1}}
\renewcommand{\H}{\mathcal{H}}            %% Hilbert space
\newcommand{\hardyon}[1]{H^2(#1)}        %% Hardy space on
\newcommand{\hideqed}{\renewcommand{\qed}{}} %% to suppress `\qed'
\newcommand{\lemref}[1]{Lemma \ref{#1}} %% Lemma X
\newcommand{\<}{\langle}

\newcommand{\N}{\mathbb{N}}              %% nonnegative integers

\newcommand{\poisson}[0]{K}               %% Poisson extension operator
\newcommand{\poissonw}[0]{K_w}            %% Poisson extension operator with weight w
\newcommand{\propref}[1]{Proposition \ref{#1}} %% Proposition X
\newcommand{\R}{\mathbb{R}}               %% real numbers

\newcommand{\remref}[1]{Remark \ref{#1}} %% Remark X
\newcommand{\secref}[1]{Section \ref{#1}} %% Section X
\newcommand{\set}[1]{\{\,#1\,\}}                %% set notation
\newcommand{\scalp}[2]{\langle {#1,#2} \rangle}      %<.,.>
\newcommand{\traceOp}{\boldsymbol{\gamma}} %% trace operator
\newcommand{\traceOpw}{\boldsymbol{\gamma}_w} %% trace operator with weight w
\newcommand{\Trw}{\Tr_\omega}             %% a Dixmier trace
\newcommand{\wt}{\widetilde}                   %% short for widetilde
\newcommand{\bb}{\begin{eqnarray}}
\newcommand{\ee}{\end{eqnarray}}
\newcommand{\eee}{\nonumber\end{eqnarray}}

\newcommand\ord{\operatorname{ord}}

\newcommand\CC{{\mathbf C}}
\newcommand\RR{{\mathbf R}}
\newcommand\dbar{\overline\partial}
\newcommand\OOm{{\overline\Omega}}
\newcommand\pOm{{\partial\Omega}}
\newcommand\cA{\mathcal A}
\newcommand\cB{\mathcal B}
\newcommand\cD{{\mathcal D}}
\newcommand\cH{\H}
\newcommand\HH{\mathbf H}
\newcommand\TTo{\mathbf T^\oplus}
\newcommand\wO{\wt\Omega}
\newcommand\wT{\wt T}
\newcommand\wH{\wt H}
\newcommand\cN{\mathcal N}
\newcounter{mnotecount}[section]
\renewcommand{\themnotecount}{\thesection.\arabic{mnotecount}}
\newcommand{\mnote}[1]%{}
{\protect{\stepcounter{mnotecount}}$^{\mbox{\footnotesize
$\bullet$\themnotecount}}$ \marginpar{\raggedright\tiny\em$\!\!\!\!\!\!\,\bullet$\themnotecount: #1} }

\def\<#1,#2>{\langle#1\,,\,#2\rangle}            %% Dirac notation
\newcommand{\norm}[1]{\left\lVert#1\right\rVert}      %% norm (C)
\newbox\ncintdbox \newbox\ncinttbox
\setbox0=\hbox{$-$} \setbox2=\hbox{$\displaystyle\int$}
\setbox\ncintdbox=\hbox{\rlap{\hbox
    to \wd2{\kern-.1em\box2\relax\hfil}}\box0\kern.1em}
\setbox0=\hbox{$\vcenter{\hrule width 4pt}$}
\setbox2=\hbox{$\textstyle\int$}
\setbox\ncinttbox=\hbox{\rlap{\hbox to \wd2{\kern-.14em\box2\relax\hfil}}\box0\kern.1em}

\begin{document}

{% beginning of local redefinition
\makeatletter\def\@fnsymbol{\@arabic}\makeatother % redefinition of \@fnsymbol
\title{Spectral triples and Toeplitz operators}
\author{M. Engli\v s \footnote{Research supported by GA~{\v CR} grant no.~201/12/G028.}\\
{\small Mathematics Institutes, Opava and Prague}
\footnote{Mathematics Institute, Silesian University in Opava, Na~Rybn\'\i\v cku~1, 74601~Opava, Czech Republic, and Mathematics Institute, Academy of Sciences, \v Zitn\' a~25, 11567~Prague~1, Czech Republic}
\\[2ex]
K. Falk and B. Iochum\\
{\small Centre de Physique Th\'eorique}
\footnote{Aix-Marseille Université, CNRS, CPT, UMR 7332, 13288 Marseille France, et Université de Toulon}
}
\date{}
\maketitle
}% end local redefinition
\vspace{2cm}

\begin{abstract}
  We give examples of spectral triples, in the sense of A. Connes, constructed using the algebra of Toeplitz operators on smoothly bounded strictly pseudoconvex domains in $\C^n$, or the star product for the Berezin--Toeplitz quantization. Our main tool is the theory of generalized Toeplitz operators on the boundary of such domains, due to Boutet de Monvel and Guillemin.
\end{abstract}

\vspace{1cm}
%%%%%%%%%%%%%%%%%%%%%%%%%%%%%%%%%%%%%%%%%%
%\centerline{\color{blue}{\bfseries \Large VERSION -- 2.1, February 13, 2014}}
%%%%%%%%%%%%%%%%%%%%%%%%%%%%%%%%%%%%%%%%%%
\newpage

%\tableofcontents
\bigskip

%%%%%%%%%%%%%%%%%%%%%%
\section{Introduction}
%%%%%%%%%%%%%%%%%%%%%%

  The purpose of this work is to construct spectral triples $(\A,\H,\DD)$ using the algebra $\A$ of Toeplitz operators acting on Bergman or Hardy spaces $\H$ on a smoothly bounded strictly pseudoconvex domain $\Omega\subset\C^n$.

  Recall that a spectral triple consists, loosely speaking, of an algebra $\A$ of operators, acting on a Hilbert space $\H$, and a certain operator $\DD$ which has bounded commutators with elements from $\A$. The notion was introduced by Connes, the model example being that of $\A$ the algebra of $C^\infty$ functions on a Riemannian manifold $M$, $\H$ the space of $L^2$-spinors and $\DD$ the Dirac operator; and a remarkable highlight of the theory is his reconstruction theorem that, in fact, every commutative spectral triple satisfying certain conditions arises in this way (up to isomorphism) \cite{ConnesReconstruction}. Spectral triples are thus quintessential for the ``noncommutative differential geometry'' program  \cite{Con94,ConnesMarcolli} and there exists an extensive literature on the subject, see e.g. the books \cite{GBVF,Khalkhali}, and the references therein.

  Our main tool are the so-called generalized Toeplitz operators, or Toeplitz operators with pseudodifferential symbols, on the boundaries of such domains, whose theory was developed by Boutet de Monvel and Guillemin \cite{BMG1981}. Upon passing from holomorphic functions on the domain to their boundary values, the generalized Toeplitz operators turn out to include also the ordinary Toeplitz operators on (weighted) Bergman spaces of the domain \cite{BM1979,Guillemin1984} (see also \cite{Howe1980} for related constructions), which have long been used in quantization on K\"ahler manifolds (Berezin and Berezin--Toeplitz quantizations) \cite{BMS,Ecmp,SchliHAB,KarSchli}.

  The set $\Psi(M)$ of classical pseudodifferential operators on a compact Riemannian manifold is an essential tool to understand the geometry of $M$. For instance, if $\A$ is the algebra $C^\infty(M)$ and $\DD\in\Psi(M)$ is of order one, both acting on $\H=L^2(M)$, then, assuming $\DD=\DD^*$, all commutators $[\DD,a]$ are bounded operators for any $a\in \A$ and $(\A,\H,\DD)$ is a spectral triple. If $\Psi^0(M)$ is the algebra of pseudodifferential operators of order less or equal to zero, one can check that $(\Psi^0(M),L^2(M),\DD)$ is also a spectral triple. An extension to non-compact manifold is possible by imposing for instance that $a(1+\DD^2)^{-1}$ is compact for any $a\in \A$. The same phenomenon occurs abstractly, since a regular spectral triple $(\A,\H,\DD)$ generates a pseudodifferential calculus $\Psi(\A)$ which is the algebra of operators $P$ on $\H$ with asymptotic expansion  $P\simeq a_q\vert\DD\vert^q+a_{q-1}\vert\DD\vert^{q-1}+\cdots$ where the $a_k$ are in the algebra generated by the $\
delta^n(a),\,n\in\N$ with $\delta(a)=[\vert\DD\vert,a]$. Then $(\Psi^0(\A),\H,\vert\DD\vert)$ will be also a spectral triple.

  Here we play around similar notions in the framework of Toeplitz operators on an open bounded domain $\Omega\subset\C^n$. By a result of Howe, the algebra $\Psi(\R^n)$ is locally isomorphic to the algebra of Toeplitz operators on the unit ball of $\C^n$ \cite{Howe1980,Taylor1984}.
  This result has been generalized by Boutet de Monvel and Guillemin \cite{BM1979,BMG1981,Guillemin1984}: the generalized Toeplitz operators on a compact manifold possessing so-called Toeplitz structure form an algebra microlocally isomorphic (modulo smoothing operators) to the algebra of pseudodifferential operators in $\R^n$. Indeed, this general framework can be brought back to a microlocal model for generalized Toeplitz operators, via a Fourier integral operator constructed modulo smoothing operators \cite{BM1974}, which sets up a bijection with the algebra of pseudodifferential operators on $\R^n$.
  In a way, the pseudodifferential operators are nothing else but Toeplitz operators in disguise \cite{Guillemin1984}.

  Of course, $\Omega$ is not compact, but we do not fall in the technicalities related to non unital spectral triples; but the price to pay is the intricate study on the role of the boundary of $\Omega$ which is of interest in complex analysis. This analysis has a long history which we intersect here only at few points: the Heisenberg algebra with its Fock and Bergman space representations and another quantization process than the Weyl one based on this Heisenberg algebra, namely the so called Berezin--Toeplitz quantization.

  We give several spectral triples for the Bergman and Hardy spaces with a natural extension for the Berezin--Toeplitz quantization based on a star product. We compute their spectral dimension, a point related to the appearance of the Dixmier trace. Such appearance was already detected in \cite{BM2008} and studied in \cite{ER2009,BEY2013,EGZ2009,EZ2010,BEY} (see also \cite{GWZ2012} for the bidisk).

  A spectral triple on the algebra $\{T_f \,\vert \, f\in C(S^1)\}\cap \Psi^0(S^1)$ acting on $\H=H^2(S^1)\oplus H^2(S^1)$ with an operator $\DD$ based on the shift and on the ``number operator'' has already been proposed in \cite{CM1993}.

  We expect that our spectral triples encode information about the (CR-)geometry of $\Omega$ or $\partial \Omega$ in much the same way as they encode the information about the Riemannian geometry of the manifold in the original model example mentioned above; we plan to treat this question in a subsequent work.

  More concretely, we recall in Section \ref{BHspaces} the role of the Poisson kernel and the trace map between the Sobolev spaces on $\partial \Omega$ and Sobolev spaces of harmonic functions on $\Omega$ which can be associated to a weight $w$ on $\overline{\Omega}$, extending known results on the analysis of Toeplitz operators of Szeg\"o type or Bergman type (see for instance \cite{E2008,BM1979}). Section \ref{FSandHA} is devoted to the Fock and Bergman representations of the Lie algebra of the Heisenberg group, especially for the model case of $\Omega$ the unit ball of $\C^n$ with the standard weights, which are used in the next section on possible Dirac-like operators on $H^2(\partial \Omega)$ or $A^2_w$. In Section \ref{Spectral triples}, we consider spectral triples with several candidates for operators $\DD$. The first one is defined by an elliptic generalized Toeplitz operator of order one; an example is the inverse of the Toeplitz operator associated with the defining function of $\Omega$, but since
  it is positive (so with trivial K-homology class), we double the
  Hilbert space and construct two classes of unitary operators compatible with the notion of generalized Toeplitz operators. A second idea is to start from the usual Dirac operator on $\R^n$ and to construct the corresponding operator acting on $\H$ through isomorphisms which involve a set of representations of the Lie algebra of the Heisenberg group. In the last section, previous results are applied to the star product arising in the Berezin--Toeplitz quantization of $\Omega$ equipped with a natural K\"ahler structure \cite{BMS,SchliHAB,KarSchli}.

%%%%%%%%%%%%%%%%%%%%%%%%%%%%%%%%%%%%%%%%%%%%%%%%%%%%%%%%%%
\section{Bergman and Hardy spaces, and Toeplitz operators}
%%%%%%%%%%%%%%%%%%%%%%%%%%%%%%%%%%%%%%%%%%%%%%%%%%%%%%%%%%
  \label{BHspaces}

  We gather in this section the preliminary material and known results.

  %%%%%%%%%%%%%%%%%%%%%%%%%%%%%%%%%%%%%%
  \subsection{Notations and definitions}
  %%%%%%%%%%%%%%%%%%%%%%%%%%%%%%%%%%%%%%
    \label{Notations and definitions}

    We consider a open strictly pseudoconvex bounded set $\Omega \subset \C^n \approx \R^{2n}$, with smooth boundary $\partial\Omega$ for $n$ a positive integer, so $\overline{\Omega}=\Omega \cup \partial \Omega$ is compact. Let $r : \overline{\Omega} \to \R^+$ be a positively signed defining function of $\Omega$, i.e.:
    \begin{align}
      \label{r}
      r\in C^\infty(\overline{\Omega}) \text{ with } r\evalat{\Omega}>0,\,  r\evalat{\partial\Omega} = 0 \text{ and } (\partial_{\mathbf{n}}r)\evalat{\partial\Omega}\neq0,
    \end{align}
    where $\partial_{\mathbf{n}}$ is the normal derivative to the boundary. The condition on the normal derivative means that $r$ behaves like the distance to the boundary near it.

    We decompose the exterior differential $d=\partial+\bar\partial$ in holomorphic and antiholomorphic parts and use $\partial_i\vc\partial_{z_i},\, \bar \partial_i\vc\partial_{\bar{z}_i}$.

    The strict pseudoconvexity guarantees that the restriction $\eta\vc\tfrac{1}{2i}(\bar \partial r -\partial r )\evalat{\partial\Omega} $  to the boundary of the one-form $\operatorname{Im}(-\partial r)$ is a contact form, i.e. $\nu \vc\eta\wedge(d\eta)^{n-1}$ is a volume element on $\partial\Omega$. Let us also consider the half-line bundle
    \begin{align}
    \label{Sigma}
      \Sigma \vc \{(x',s\,\eta_{x'}) \in T^*(\partial\Omega),\, s>0 \},
    \end{align}
    and a positive weight function $w$ on $\overline{\Omega}$ which is decomposed the following way:
    \begin{align}
      \label{w}
      w \cv r^{m_w} \,g_w,  \text{ where }m_w\in\R \text{ with }m_w>-1 \text{ and }\,g_w\in\Coo(\overline{\Omega})\text{ with }g_w\evalat{\partial\Omega}>0.
    \end{align}

    As Hermitian structure, we choose the Euclidean one on $\C^n\simeq \R^{2n}$, so $\norm{dz_j}=\sqrt{2}$ and $\norm{\partial r}=\sqrt{2} \norm{\eta}$.

    In the following, all pseudodifferential operators are classical and considered in $OPS^d_{1,0}$, for some $d\in\R$, using Hörmander notation (see \cite[Definition 7.8.1]{HorI}). For any pseudodifferential operator $Q$, we will denote by $\sigma(Q)$ its principal symbol.

    \begin{definition}\label{bergspace}
      The weighted Bergman space is
      \begin{align*}
        A^2_w(\Omega) \vc L^2_{hol}(\Omega,w) \vc \{ f\in L^2(\Omega,w\,d\mu), \, f \text{ is holomorphic on } \Omega \},
      \end{align*}
      endowed with the norm derived from the weighted scalar product $\scalp{f}{g}_{w} \vc \int_\Omega f \,\bar{g}\, w\,d\mu$,
      where $\mu$ is the Lebesgue measure. \\
      The Bergman space $A^2(\Omega)$ is the unweighted Bergman space where $w = r^0=1$.
    \end{definition}

    When $\Omega = \mathbb{B}^n$ and the defining function $r$ is radial, i.e. $r(z)=r(\abs{z})$, and for a weight $w=r^{m_w}$, $m_w\in\N$, we have the following orthonormal basis for $A^2_w(\mathbb{B}^n)$ (see \cite[Corollary 2.5]{GKV2003}):
    \begin{align} \label{valpharad}
      v_\alpha^w(z) = b_{\alpha} \,z^\alpha \vc  \big[\int_{\mathbb{B}^n} z^{\alpha} \overline{z}^{\alpha} \,w(\abs{z})\, d\mu(z) \big]^{-1/2}\, z^\alpha \,,\,\,\ \alpha \in \N^n.
    \end{align}
    In particular, an orthonormal basis of the unweighted Bergman space $\berg$ is given by the family
    \begin{align}\label{valpha}
      v_\alpha(z) = b_\alpha \,z^\alpha \vc \big(\tfrac{(\abs{\alpha}+n)!}{n!\,\alpha!\,\mu(\mathbb B^n)}\big)^{1/2} \, z^{\alpha},
    \end{align}
    (see \cite[Lemma 1.11]{Zhu2005} for more details) where for a multiindex $\alpha \in \N^n$, we set $\alpha! \vc \prod_{k=0}^n \alpha_k!$ and $\abs{\alpha} \vc \sum_{k=0}^n \alpha_k$.

    We use the standard definition of Sobolev spaces of order $s\in\R$ on a subset of $\R^n$ or its boundary, and on $\Omega \subset\C^n$ or $\partial \Omega$, whose construction is given in \cite[Chapitre 1]{Magenes1961}, and also in \cite[Appendix]{Grubb1996}. We denote them by $W^{s}(\Omega)$ and $W^{s}(\partial\Omega)$, respectively.
    We assume that the norms in these spaces have been chosen so that $W^0(\Omega)=L^2(\Omega,d\mu)$ with $\mu$ as above, and $W^0(\partial\Omega)=L^2(\partial\Omega)$ with respect to some smooth volume element on $\partial\Omega$, absolutely continuous with respect to the surface measure (for instance, $\nu=\eta\wedge(d\eta)^{n-1}$), which we fix from now on.

    \begin{definition}
      For $s \geq 0$, the holomorphic (resp. harmonic) Sobolev space on $\Omega$ of order $s$ is defined by
      \begin{align*}
        W_{hol}^{s}(\Omega) \,\big(resp. \, W_{harm}^{s}(\Omega)\big) \vc \{  f\in W^{s}(\Omega) \text{, } f \, \text{ is holomorphic (resp. harmonic) on }\Omega \}.
      \end{align*}
      Thus, $W_{hol}^{0}(\Omega) = \bergon{\Omega}$.\\
      The set of harmonic function in $L^2(\Omega,\, w\, d\mu)$ is denoted $L^2_{harm}(\Omega,\,w)$.
    \end{definition}

    \begin{definition}
      The Poisson operator $K$ is the harmonic extension operator which solves the Dirichlet problem: $\Delta\,K u = 0\,\text{ on }\Omega, \,\, K u\evalat{\partial\Omega} = u$, where $\Delta = \partial\overline{\partial}$ is the complex Laplacian.
    \end{definition}

    Thus $K$ acts from functions on $\partial\Omega$ into harmonic functions on $\Omega$ and by elliptic regularity theory (see \cite{Magenes1961}), $K$ extends to a continuous map from $W^s(\partial\Omega)$ onto $W_{harm}^{s+1/2}(\Omega,w)$, for all $s\in\R$. In particular $K:C^\infty(\partial\Omega)\to C_{harm}^\infty(\overline{\Omega})$. We denote by $K_w$ the operator $K$ considered as acting from $L^2(\partial\Omega)$ into $L^2(\Omega,w\,d\mu)$, and by $K^*_w$ its Hilbert space adjoint. For an arbitrary weight $w$, a simple computation shows $K^*_w$ is related to $K(=K_1)$ through
    $$
      K^*_w u = K^*(wu) .
    $$
    In particular, $K_w^*K_w=K^*wK$. \\Note that $K_w$ is injective since $0 = \poissonw u \Rightarrow \poissonw u\evalat{\partial\Omega} = 0 \Leftrightarrow u=0$. The operator $K^*$ acts continuously from $W^s(\Omega)$ into $W^{s+1/2}(\partial\Omega)$, for all $s\in\mathbb R$. Thus $K: C^\infty(\partial\Omega) \to C_{harm}^\infty(\overline{\Omega})$.

    We consider the operator
      \begin{align}
        \label{Lambdaw}
        \Lambda_w \vc \poissonw ^*\poissonw = \poisson ^* \, w \, \poisson.
      \end{align}
    Actually $\Lambda_w$ is an elliptic and selfadjoint pseudodifferential operator of order $-(m_w+1)$ on $\partial\Omega$ (hence compact) with principal symbol (see \cite{BM1979})
    \begin{align}
      \label{princ}
      \sigma(\Lambda_w)(x',\xi') = 2^{-1}\Gamma(m_w+1) \,g_w(x') \|\eta_{x'}\|^{m_w}\,\norm{\xi'}^{-(m_w+1)}  \quad x'\in\partial\Omega, \,\xi' \in T^*\partial\Omega,
    \end{align}
    so, when $m_w\in \N$,
    \begin{align}
      \label{princ1}
      \sigma(\Lambda_w)(x',\xi')= 2^{-(m_w+1)}\,(\partial_{\mathbf{n}}^{m_w}w)(x')\,\norm{\xi'}^{-(m_w+1)}, \quad x'\in\partial\Omega, \,\xi' \in T^*\partial\Omega.
      \end{align}
    This is actually a subject of the extensive theory of calculus of boundary pseudodifferential operators due to Boutet de Monvel \cite{BdMActa}.

    In particular, $\Lambda_w$ acts continuously from $W^s(\partial\Omega)$ into $W^{s+m_w+1}(\partial\Omega)$, for any $s\in\R$. $\Lambda_w$ is an injection since for $u\in \Ker(\Lambda_w)$ and using the injectivity of $\poissonw$, we have $0=\scalp{\Lambda_w u}{u} = \norm{\poissonw u}^2$.
    The inverse operator $\Lambda_w^{-1}$ is well defined on $\text{Ran}(\poissonw ^*)$, thus we have
    \begin{align}
    \label{KLambdaK*}
      \Lambda_w^{-1}\,\poissonw^* \poissonw = \bbbone_{L^2(\partial\Omega)} \, \text{ and} \,\, \poissonw \, \Lambda_w^{-1}\,\poissonw^*  = \mathbf{\Pi}_{w,harm}\,,
    \end{align}
    where  $\mathbf{\Pi}_{w,harm}$ is the orthogonal projection from $L^{2}(\Omega,w)$ onto $L^{2}_{harm}(\Omega,w)$: the first equality is direct. Applying $K_w$ on both side of it, we deduce that $\poissonw \Lambda_w^{-1} \poissonw^*$ is the identity on $\overline{\text{Ran}(\poissonw)}$ which is the closure of $W^{1/2}_{harm}(\Omega)$ in $L^2(\Omega,w)$, i.e. $L^2_{harm}(\Omega,w)$. Moreover, $\poissonw \Lambda_w^{-1} \poissonw^*$ vanishes on $\overline{\text{Ran}(\poisson_w)}^\perp = \text{Ker}(\poissonw^{*})$, and we get the second equality of \eqref{KLambdaK*}.\\
    As a bounded operator, $\poissonw$ has the polar decomposition $\poissonw  \cv U_w (\poissonw ^*\poissonw )^{1/2} = U_w \Lambda_w^{1/2}$, where $U_w$ is a unitary from $L^{2}(\partial\Omega)$ onto $\overline{\text{Ran}(\poissonw )} = L^2_{harm}(\Omega,w)$: $U_w^*U_w$ maps $L^2(\partial\Omega)$ to itself since
    \begin{align*}
      U_w^*U_w=\Lambda_w^{-1/2}\,\poissonw ^* \poissonw \, \Lambda_w^{-1/2} =\Lambda_w^{-1/2}\Lambda_w\Lambda_w^{-1/2}=\bbbone_{L^2(\partial\Omega)},
    \end{align*}
    while $U_wU_w^*  = \poissonw \Lambda_w^{-1}\poissonw ^*\evalat{L^2_{harm}(\Omega,w)}=\bbbone_{L^2_{harm}(\Omega,w)}$ by \eqref{KLambdaK*}.

    \begin{definition}
      The trace operator $\traceOpw  : L^2(\Omega,w) \to L^2(\partial\Omega)$ is defined by
      \begin{align*}
        \traceOpw \vc \Lambda_w^{-1}\poissonw ^*\,.
      \end{align*}
    \end{definition}
    In particular, \eqref{KLambdaK*} gives
    \begin{align*}
      & \poissonw \,\traceOpw\evalat{L^2_{harm}(\Omega,w)} = \bbbone_{L^2_{harm}(\Omega,w)}\,\,\text{and }\,\, \traceOpw \,\poissonw  = \bbbone_{L^2(\partial\Omega)}.
    \end{align*}

    The operator $\traceOpw$ is thus a left inverse of $\poissonw$, so it takes the boundary value of any function $f$ in $L^2_{harm}(\Omega,w)$. Again, the index $w$ is just here to recall that the operator is defined on a weighted Hilbert space. The operator $\traceOpw$ extends continuously to $\traceOpw : W^s_{harm}(\Omega) \to W^{s-1/2}(\partial\Omega)$ for any $s\in\R$.
    \bigskip

    We now define the following spaces:
    \begin{definition}
      The holomorphic Sobolev space on $\partial\Omega$ of order $s\in\R$ is
      \begin{align*}
        W_{hol}^{s}(\partial\Omega) \vc \{  u\in W^{s}(\partial\Omega) \text{, } \poisson u \, \text{ is holomorphic on }\Omega \},
      \end{align*}
      and the Hardy space is
      \begin{align*}
        H^2 \vc H^2(\partial\Omega) \vc W_{hol}^{0}(\partial\Omega),
      \end{align*}
      with the usual norm on $L^2(\partial\Omega)$.
    \end{definition}
    Note that the Hardy space is the closure of $C^{\infty}_{hol}(\partial\Omega)$ in $L^2(\partial\Omega)$.

    We will use two types of Toeplitz operators using bold letters to refer to operators acting on Hilbert spaces defined over the domain $\Omega$ whereas the regular roman ones concern those over its boundary $\partial\Omega$:

    \begin{definition}
      If $u\in \Coo(\partial\Omega)$, the Toeplitz operator $T_u : H^2(\partial\Omega) \to H^2(\partial\Omega)$ is
      \begin{align*}
        T_u \vc \Pi \, M_u,
      \end{align*}
      where $\Pi : L^2(\partial\Omega) \to H^2(\partial\Omega)$ is the Szeg\"o projection and $M_u$ (or just $u$) is the multiplication operator by $u$.
      \\
      For $f\in \Coo(\overline{\Omega})$, the Toeplitz operator $\mathbf{T}_f : \bergonw{w}{\Omega} \to \bergonw{w}{\Omega}$ is defined as
      \begin{align*}
        \mathbf{T}_f \vc \mathbf{\Pi}_w \, \mathbf{M}_f,
      \end{align*}
      where $\mathbf{\Pi}_w : L^2(\Omega,w) \to \bergonw{w}{\Omega}$ is the orthogonal projection onto the space of holomorphic functions in $L^2(\Omega,w)$ and $\mathbf{M}_f$ is the multiplication by $f$.
    \end{definition}
    For the Hardy (resp. Bergman) case, we have
    \begin{align*}
      u \to T_u \,\text{ is linear,}\,\, T^*_u = T_{\overline{u}}, \,\,T_{1} = \bbbone, \,\, \norm{T_u} \leq \norm{u}_{\infty}, \,\,(\text{resp. idem for } \mathbf{T}_f).
    \end{align*}

    \begin{remark}
    \label{reminvTf}
      For any strictly positive function $u$ in $L^{\infty}(\partial\Omega)$, $T_u$ is a selfadjoint and positive definite operator on $H^2(\partial\Omega)$ since $\scalp{T_u v}{v} =\int_{\partial \Omega}u \vert v\vert^2>0$ for any $v\neq 0$.
      In particular, it is an injection, so there exists an unbounded inverse $T_u^{-1}$, which is densely defined on $H^2(\partial\Omega)$. The same is true in the Bergman case for $\mathbf{T}_f$, $f\in L^{\infty}(\Omega)$.
    \end{remark}
    Boutet de Monvel and Guillemin studied in \cite{BMG1981} (see also \cite{BM1979}) a more general notion of Toeplitz operators acting on Hardy spaces:
    \begin{definition}
      \label{defGTO}
      For a pseudodifferential operator  $P$ on $L^2(\partial\Omega)$ of order $m\in\R$, let $T_P$ be the generalized Toeplitz operator (GTO): $W_{hol}^{m}(\partial\Omega) \to H^2(\partial\Omega)$ defined by
      \begin{align*}
        T_P \vc \Pi \,P\evalat{W_{hol}^{m}(\partial\Omega)}.
      \end{align*}
    \end{definition}
    One can alternatively extend the definition of $T_P: W^{m}(\partial\Omega) \to H^2(\partial\Omega)$ by taking $T_P = \Pi \,P\, \Pi$.

    It often happens that $T_P = T_Q$ with $P\neq Q$. However, the restriction of the principal symbol $\sigma(P)$ of $P$ to $\Sigma$ is always determined uniquely: when $T_P = T_Q$ then either $\ord(P)=\ord(Q)$ and in that case $\sigma(P)\evalat{\Sigma} = \sigma(Q)\evalat{\Sigma}$; or, for instance, $\ord(P)>\ord(Q)$ and in that case $\sigma(P)\evalat{\Sigma} = 0$. Therefore, the following quantities are well defined:
    \begin{definition}
    \label{deforder}
      The order and the principal symbol of a GTO $T_P$ are respectively
      \begin{align*}
        &\, \ord(T_P) \vc  \inf\{ \ord(Q),\,T_P = T_Q \}, \\
        &\, \sigma(T_P) \vc  \sigma(Q)\evalat{\Sigma}, \, \text{for any $Q$ such that $T_Q=T_P$ and $\ord(Q)=\ord(T_P)$}.
      \end{align*}
      The order can be $-\infty$, in which case the symbol is not defined.
    \end{definition}
    As shown in \cite{BMG1981}, for any GTO $T_P$ there exists a pseudodifferential operator $Q$ such that $T_P=T_Q$ and $\commut{Q}{\Pi} = 0$. As a consequence, the GTO's form an algebra: if $P,\,Q$ are two pseudodifferential operators, there exists another pseudodifferential operator $R$ such that $T_P\,T_Q=T_R$. We have also the usual properties
    \begin{align*}
      & \ord(T_P\,T_Q) = \ord(T_P) + \ord(T_Q),\\
      & \, \sigma(T_P\,T_Q) = \sigma(T_P)\,\sigma(T_Q).
    \end{align*}
    Moreover, a GTO $T_P$ of order $m$ maps continuously holomorphic Sobolev spaces, namely
    \begin{align*}
      T_P : W_{hol}^{s+m}(\partial\Omega) \to W_{hol}^{s}(\partial\Omega), \, \text{for any $s\in\R$},
    \end{align*}
    because $\Pi$ is (or rather extends to) a continuous map from $W^s(\partial\Omega)$ onto $W^s_{hol}(\partial\Omega)$ for any real number $s$.

    A GTO is said to be elliptic if its principal symbol does not vanish. Like classical pseudodifferential operators, an elliptic GTO $T_P$ of order $m$ admits a parametrix $T_Q$ which is a GTO of order $-m$, verifying
    \begin{align*}
      \,\sigma(T_Q) =\sigma(T_P)^{-1},\quad T_P\,T_Q \sim \bbbone \,\,\, \text{ and } \,\, \,T_Q\,T_P \sim \bbbone.
    \end{align*}
    Here and below $A\sim B$ means that $A-B$ is a smoothing operator (i.e. of order $-\infty$, or equivalently having Schwartz kernel in $C^\infty(\partial\Omega\times\partial\Omega)$).

    Finally, if $T_P$ is elliptic of order $m\neq 0$, positive and selfadjoint as an operator on $H^2(\partial\Omega)$, with $\sigma(T_P)>0$, then the power $T_P^s$, $s\in\C$ (in the sense of the spectral theorem) is a GTO of order $ms$. In particular, for $s=-1$, the inverse of $T_P$ is a GTO of order $-m$ (see \cite[Proposition 16]{E2008} for the details).

    Let $\mathbf{P}$ be a differential operator on $\C^n$ of order $d \in \mathbb N$ of the form
    \begin{equation}
      \label{diffonCn}
      \mathbf P=\sum_{|\nu|\leq d} a_\nu(x) \,r(x)^j\, \partial^\nu +\sum_{|\nu'|\leq d} b_{\nu'}(x) \,r(x)^{j'}\, \overline{\partial}\,^{\nu'}
    \end{equation}
    for some $d\in\N$, $j,j'\in \R^+$, $\nu,\nu'\in \N^{n}$ and some functions $a_\nu,b_{\nu'}\in C^\infty(\overline\Omega)$.

    We can generalize the definition of $\Lambda_w$ and construct the operator
    \begin{align}
      \label{K*wPK}
      \Lambda_{w\mathbf{P}} \vc K_w^* \mathbf{P} K_w=K^*w\mathbf P K
    \end{align}
    acting on the boundary $\partial\Omega$. By Boutet de Monvel's theory \cite{BdMActa,Grubb1996,RS}, $\Lambda_{w\mathbf P}$ is again a pseudodifferential operator on the boundary and the following proposition gives a formula for its principal symbol.

    \begin{prop}
    \label{propTPpsdo}
      The operator $\Lambda_{w\mathbf P}$ is a pseudodifferential operator on the boundary $\partial \Omega$ and $T_{\Lambda_{w\mathbf P}}$ is a GTO of order $d-(m_w+1+j)$ with principal symbol
      \begin{align}
        \label{princP}
        \sigma(T_{\Lambda_{w\mathbf{P}}})(x',\xi') = \tfrac{(-1)^d\,\Gamma(m_w+1+j)}{2 \norm{\xi'}^{-d+m_w+1+j}}\,g_w(x')\, \|\eta_{x'}\|^{-d+m_w+j}  \sum_{|\nu|=d} a_\nu(x')\, \prod_{k=1}^n (\partial_k r)^{\nu_k}(x')
      \end{align}
    (when \eqref{princP} vanishes, $T_{\Lambda_{w\mathbf{P}}}$ is in fact of lower order).
    \end{prop}

    \begin{proof}
      For $k\in\set{1,\cdots,n}$, define the tangential operators $Z_k$ and $\overline{Z}_k$ on $\partial\Omega$ by
      \begin{align}
        \label{Z_k}
        Z_k \vc \gamma\,\partial_{k} \,K, \quad \overline{Z}_k \vc \gamma\,\overline{\partial}_{k} \,K,
      \end{align}
      As $K^*w\mathbf P K=\sum_\nu K^* a_\nu r^j w \partial^\nu K + \sum_{\nu'} K^* b_{\nu'} r^{j'} w \overline{\partial}^{\nu'} K = \sum_\nu\Lambda_{a_\nu r^j w} Z^\nu + \sum_{\nu'}\Lambda_{b_{\nu'} r^{j'} w} \overline{Z}^{\nu'}$ with $Z=\gamma\partial K$ and the same for $\overline{Z}$, we see from \eqref{princ} that indeed $K^*w\mathbf P K$ is a pseudodifferential operator on $\partial\Omega$ of order $d-(m_w+j+1)$ (or less if there are some cancellations in the summation on $\nu$). Since $T_{\Lambda_{w\mathbf{P}}} = \Pi \, \Lambda_{w\mathbf{P}} \, \Pi$, we have
      \begin{align*}
        \sigma(T_{\Lambda_{w\mathbf{P}}})(x',\xi') = \sigma(\Lambda_{w\mathbf{P}}\evalat{H^2})(x',\tfrac{\norm{\xi'}}{\norm{\eta_{x'}}}\,\eta_{x'})=\sigma(\sum_\nu\Lambda_{a_\nu r^j w} Z^\nu)(x',\tfrac{\norm{\xi'}}{\norm{\eta_{x'}}}\,\eta_{x'})
      \end{align*}
      By a direct computation, $\sigma(Z_k)(x',\xi')= i \scalp{\xi'}{Z_k} =-\tfrac{\norm{\xi'}}{\norm{\eta_{x'}}} \partial_k r$ (see also \cite[p. 1440]{E2008}) and from \eqref{princ}, we have
      \begin{align}
        \label{ttmmpp}
        \sigma(T_{\Lambda_{w\mathbf{P}}})(x',\xi') = \tfrac{(-1)^d\,\Gamma(m_w+1+j)}{2 \norm{\xi'}^{m_w+1+j}}\,g_w(x')\,\|\eta_{x'}\|^{m_w+j} \,\sum_{\abs{\nu}=d}\,a_\nu(x') \,\prod_k\sigma(Z_k)^{\nu_k}(x',\xi'),
        \tag*{\qed}
      \end{align}
      \hideqed
      so the result follows.
    \end{proof}

    We remark that equation \eqref{princP} is still valid when $m_w \in \C$ with $\text{Re}(m_w)>-1$. Also, when $m_w\in \N$, using $\norm{\partial r}  =\sqrt{2} \norm{\eta_{x'}}$, the right-hand side can be written as
    \begin{align*}
      \tfrac{(-1)^d\,\Gamma(m_w+1+j)}{2^{-d+m_w/2+j+1}\Gamma(m_w+1)^{-d+j}\,\norm{\xi'}^{-d+m_w+1+j}} \, \big(\partial_{\mathbf{n}}^{(m_w)}w \big)^{-d+j}(x')\, \, \sum_{|\nu|=d} a_\nu(x')\, \prod_{k=1}^n (\partial_k r)^{\nu_k}(x').
    \end{align*}

  %%%%%%%%%%%%%%%%%%%%%%%%%%%%%%%%%%%%%%%%%%%%%%%%%%%%%%%%%%
  \subsection{\texorpdfstring{Links between the spaces on $\Omega$ and $\partial \Omega$}{}}
  %%%%%%%%%%%%%%%%%%%%%%%%%%%%%%%%%%%%%%%%%%%%%%%%%%%%%%%%%%

    The operator $T_{\Lambda_w}$ exists as a positive, elliptic and compact GTO of order $-(m_w+1)$ on $L^2(\partial\Omega)$ and maps continuously $W_{hol}^s(\partial\Omega)$ into $W_{hol}^{s+m_w +1}(\partial\Omega)$, for any $s\in\R$. \\
    Let $u\in \Ker(T_{\Lambda_w}) \subset W^s_{hol}(\partial\Omega)$ for a certain $s\in\R$, then
    \begin{align*}
      0 = \scalp{T_{\Lambda_w}u}{u}_{W^s_{hol}(\partial\Omega)} = \scalp{\Pi\Lambda_w u}{ u}_{W^s_{hol}(\partial\Omega)} = \scalp{\Lambda_w u}{\Pi u}_{W^s_{hol}(\partial\Omega)}.
    \end{align*}
    Since $\Pi u = u$ we get, using the injectivity of $\Lambda_w$
    \begin{align*}
      0 = \scalp{\Lambda_w u}{u}_{W^s_{hol}(\partial\Omega)} = \Vert\Lambda_w^{1/2} u \Vert^2 \Rightarrow u = 0.
    \end{align*}
    Thus, for any $s \in \R$, the inverse operator $T_{\Lambda_w}^{-1}$ exists from $\text{Ran}(T_{\Lambda_w}) = W_{hol}^{s+m_w+1}(\partial\Omega)$ onto $W_{hol}^{s}(\partial\Omega)$.

    For completeness, we give the proof of the following result from \cite[Theorem 4]{E2010}.

    \begin{prop}
      \label{WT-W}
      Let $T$ be a positive selfadjoint operator on $H^2(\partial\Omega)$ such that $T\sim T_P$, where $P$ is an elliptic pseudodifferential operator of order $s\in\R$ such that $\sigma(T_P)>0$. \\
      Let $W^T_{hol}(\partial\Omega)$ be the completion of $C^{\infty}_{hol}(\partial\Omega)$ with respect to the norm
      $$
        \norm{u}_T^2 \vc \scalp{Tu}{u}_{H^2}.
      $$
      Then, we have
      \begin{align*}
        W^{T}_{hol}(\partial\Omega) = W^{\ord(T_P)/2}_{hol}(\partial\Omega).
      \end{align*}
    \end{prop}

    \begin{proof}
      We may assume that $P$ commutes with $\Pi$. The equivalence between $T$ and $ T_P$ induces $T^{1/2}\sim T_P^{1/2}\sim T_{P^{1/2}} = \Pi\, P^{1/2}\evalat{H^2} = P^{1/2}\evalat{H^2}$. Since $P^{1/2}$ is elliptic, the GTO $T_{P^{1/2}}$ admits a parametrix, so is Fredholm. As a consequence $T^{1/2}$ is a positive, so an injective Fredholm operator, thus an isomorphism from $W^{\ord(T)/2}(\partial\Omega)$ onto $H^2(\partial\Omega)$. Now if $u\in W^{\ord(T)/2}_{hol}(\partial\Omega)$, $T^{1/2}u$ belongs to $H^2(\partial\Omega)$. So we have the finite quantities
      \begin{align*}
        \Vert T^{1/2}u\Vert_{H^2}^2 = \scalp{Tu}{u}_{H^2} = \norm{u}_{W^{T}_{hol}}^2,
      \end{align*}
      which proves the equality $W^{T}_{hol}(\partial\Omega)= W^{\ord(T)/2}_{hol}(\partial\Omega)$.
    \end{proof}

    \begin{prop}
    \label{bijK}
      The operator $\poissonw$ maps bijectively the space $W_{hol}^{T_{\Lambda_w}}(\partial\Omega)$ onto $A^2_w(\Omega)$.
    \end{prop}

    \begin{proof}
      Let $f\in C^\infty_{hol}(\overline\Omega)\subset A^2_w(\Omega)$ and $u=\traceOpw f\in C^\infty_{hol}(\overline\Omega)$. Then
      \begin{align*}
        \norm{f}_{A^2_w}^2 = & \, \scalp{\poissonw u}{ \poissonw u}_{L^2(\Omega)} = \scalp{\Lambda_w u }{u}_{L^2(\partial\Omega)} = \scalp{\Pi\Lambda_w u }{u}_{L^2(\partial\Omega)} = \scalp{T_{\Lambda_w} u }{u}_{L^2(\partial\Omega)}.
      \end{align*}
      Thus $K_w$ is an isometry of $W_{hol}^{T_{\Lambda_w}}(\partial\Omega)$ onto the completion of $C^\infty_{hol}(\overline\Omega)$ in $A^2_w$. \\
      In a manner completely similar to \eqref{KLambdaK*} (see \cite{EZ2010} for details), we get
      \begin{align}
        \label{mathbf{Pi}_w}
        \mathbf{\Pi}_w = K_w\,\Pi\, T_{\Lambda_w}^{-1}\, \Pi\, K^*_w.
      \end{align}
      Since  $C^\infty_{hol}(\overline\Omega)$ is dense in $A^2_w$ (just note that $C^\infty(\overline\Omega)$ is dense in $L^2(\Omega,w)$, while the weighted Bergman projection $\mathbf{\Pi}_w$ maps each $W^s_{hol}(\Omega)$, and, hence, $C^\infty(\overline\Omega)$ into itself), the claim follows.
    \end{proof}

    From the fact that $K_w$ is an isomorphism of $W^s_{hol}(\partial\Omega)$ onto $W^{s+1/2}_{hol}(\Omega)$ $\forall s\in\R$, we also see that $A^2_w(\Omega)=K_w W^{-(m_w+1)/2}(\partial\Omega)=W^{-m_w/2}_{hol}(\Omega)$ and $\traceOpw$ is an isomorphism of $A^2_w$ onto $W^{-(m_w+1)/2}(\partial\Omega)$.

    As we already said, $T_{\Lambda_w}^{1/2}$ is an isomorphism of $W^s_{hol}(\partial\Omega)$ onto $W^{s+\frac{m_w+1}2}(\partial\Omega)$ for all $s\in\R$, with equivalent norms. As a consequence,
    \begin{lemma}
      The operator
      \begin{align}
        \label{V_w}
        V_w\vc \poissonw \, T_{\Lambda_w}^{-1/2}
        \text{ is a unitary which maps }H^2(\partial\Omega) \text{ onto }\bergonw{w}{\Omega}.
      \end{align}
    \end{lemma}
    \begin{proof}
      We have $V_w^*V_w =  T_{\Lambda_w}^{-1/2} \poissonw ^*\poissonw \, T_{\Lambda_w}^{-1/2} = T_{\Lambda_w}^{-1/2} \, T_{\Lambda_w} \, T_{\Lambda_w}^{-1/2} = \bbbone_{H^2}$. \\Similarly, $V_wV_w^*=\bbbone_{\bergonw{w}{\Omega}}$, see \eqref{mathbf{Pi}_w}.
    \end{proof}

    Now we identify a Toeplitz operator $\mathbf{T}_f$ on $\bergonw{w}{\Omega}$ with generalized Toeplitz operators acting on $\hardyon{\partial\Omega}$ via $\gamma_w$ and $K_w$ or $V_w$ and $V^*_w$:

    \begin{prop}
      \label{gammaTK}
      For $f\in\Coo(\overline{\Omega})$, we have
      \begin{equation}
        \label{Tf=VTTTV*}
        \begin{aligned}
          & \traceOp_w \, \mathbf{T}_f  \, \poissonw  = T_{\Lambda_w}^{-1}\,T_{\Lambda_{wf}} \quad\text{on } W_{hol}^{-(m_w+1)/2}(\partial\Omega),\\
          & \mathbf{T}_f  = V_w \, T_{\Lambda_w}^{-1/2}\,T_{\Lambda_{wf}} \, T_{\Lambda_w}^{-1/2}\,  V_w^* \quad\text{on } \bergonw{w}{\Omega}.
        \end{aligned}
      \end{equation}
    \end{prop}

    \begin{proof}
      For any $u$ an $v$ in $W_{hol}^{-(m_w+1)/2}(\partial\Omega)$, we get
      \begin{align*}
        \scalp{\mathbf{T}_f \, \poissonw u}{\poissonw v}_{A^2_w(\Omega)}  &=  \scalp{\mathbf{\Pi}_w \, f \, \poissonw u}{\poissonw v}_{A^2_w(\Omega)} = \scalp{ f \, \poissonw u}{\mathbf{\Pi}_w \, \poisson v}_{A^2_w(\Omega)} = \scalp{ f \, \poissonw u}{\poissonw v}_{A^2_w(\Omega)} \\
        &= \scalp{ wf \, \poisson u}{\poisson v}_{L^2(\Omega)} = \scalp{ (\poisson ^* \, wf \, \poisson ) u}{v}_{H^2(\partial\Omega)} = \scalp{ \Lambda_{wf} u}{\Pi v}_{H^2(\partial\Omega)} \\
        &= \scalp{T_{\Lambda_{wf}}u}{v}_{H^2(\partial\Omega)}  = \scalp{\poissonw \,T_{\Lambda_w}^{-1}\,T_{\Lambda_{wf}}\,u}{\poissonw v}_{H^2(\partial\Omega)}.
      \end{align*}
      Thus $\mathbf{T}_f \, \poissonw  = \poissonw \,T_{\Lambda_w}^{-1}\,T_{\Lambda_{wf}}$ on $W_{hol}^{-(m_w+1)/2}(\partial\Omega)$, hence $\traceOpw \mathbf{T}_f \poissonw = T_{\Lambda_w}^{-1}\,T_{\Lambda_{wf}}$.\\
      Finally, we get
      $V_w^*\, \mathbf{T}_f \,V_w =V_w^* \,(K_w\gamma_w) \,\mathbf T_f\, (K_w\gamma_w) \,V_w =T_{\Lambda_w}^{-1/2} \, T_{\Lambda_{wf}} \, T_{\Lambda_w}^{-1/2}$.
    \end{proof}
    From the GTO's theory and the mapping properties of $\poisson_w$ and $\traceOpw$, we see that the right-hand side in \eqref{Tf=VTTTV*} extends to a bounded operator on any $W^s_{hol}(\Omega)$, hence the left-hand side enjoys the same property.

    To any differential operator $\mathbf{P}$ with coefficients as in \eqref{diffonCn}, we can associate the ``Toeplitz operator with symbol $\mathbf P$''
    \begin{equation}\label{upsPX}
      \mathbf T_{\mathbf P} := \boldsymbol\Pi_w \mathbf P \quad\text{acting on } \bergonw{w}{\Omega}.
    \end{equation}

    \begin{lemma}
      \label{TP=TP*}
      For $\mathbf P$ as above, we have
      \begin{equation}
        \label {upsP}
        \begin{aligned}
          & \gamma_w \mathbf{T}_{\mathbf{P}} \,K_w =  T_{\Lambda_w}^{-1}\,T_{\Lambda_{w\mathbf{P}}} \,  \quad\text{on } W_{hol}^{-(m_w+1)/2}(\partial\Omega),\\
          &  \mathbf{T}_{\mathbf{P}} =  V_w \, T_{\Lambda_w}^{-1/2}\,T_{\Lambda_{w\mathbf{P}}} \, T_{\Lambda_w}^{-1/2}\,  V_w\,\hspace{-0.2cm}^* \,\quad\text{on } \bergonw{w}{\Omega}.
        \end{aligned}
      \end{equation}

      Moreover, $ \mathbf T_{\mathbf P}$ is selfadjoint on $\bergonw{w}{\Omega}$ when $\mathbf{P}$ has a selfadjoint extension on $L^2(\Omega,w)$.
    \end{lemma}

    \begin{proof}
      Similar calculation as in the proof of \eqref{Tf=VTTTV*} shows \eqref{upsP}. \\
      We only needs to prove that $(T_{\Lambda_{w\mathbf{P}}})$ is selfadjoint, which follows from \eqref{K*wPK}: for $u,\,v\in H^2(\partial \Omega)$,
      \begin{align*}
        \langle\, (T_{\Lambda_{w\mathbf{P}}})^* \,u,\,v \rangle_{H^2(\partial \Omega)}&=\langle u, K^*w\mathbf PK\,v \rangle_{H^2(\partial \Omega)}=\langle K u,\,\mathbf PK v \rangle_{L^2(\Omega,w)}=\langle \mathbf PK u,\,K v \rangle_{L^2(\Omega,w)}\\
        &=\langle w\mathbf PK u,\,K v \rangle_{L^2(\Omega)}=\langle Kw\mathbf PK\,u,v \rangle_{H^2(\partial \Omega)}=\langle T_{\Lambda_{w\mathbf{P}}} \,u,\,v \rangle_{H^2(\partial \Omega)}.
      \end{align*}
    \end{proof}

    \begin{remark}
      \label{TPselfadjoint}
      { \rm
      An interesting example of selfadjoint operator $\mathbf T_{\mathbf P}$, where $\mathbf P$ is not selfadjoint on $L^2(\Omega,w)$ is given by the ``weighted normal derivative'' operator
      \begin{align*}
        \mathbf P_{w} \vc \sum_{j=1}^n \tfrac{\overline{\partial_j}(rw)}{w}\,\partial_j \,;
      \end{align*}
      note that $\partial_j(rw)/w$ is smooth up to the boundary. Using Stokes' formula, for $f,\,g$ in $\bergonw{w}{\Omega}$,
      $$
        \int_\Omega d\mu\,\partial_j(rw\,f\overline g) =- \int_{\partial\Omega} d\sigma\,rw\,f\overline g\, \tfrac{\partial_j r}{2\norm{\partial r}} = 0
      $$
      since $rw=0$ on $\partial \Omega$ (here $d\sigma$ is the ordinary surface measure on $\partial\Omega$). Applying Leibniz rule to the LHS gives
      $$
        \int_\Omega wd\mu\,\tfrac{\partial_j(rw)}{w}\,f\overline g   = - \int_\Omega wd\mu\,(r\partial_j )f\,\overline g \quad\text{hence }\,\,\mathbf T_{\partial _j (rw)/w} + \mathbf T_{r\partial_j} = 0
      $$
      where the Toeplitz operators act on $A_w^2(\Omega)$. Since $\mathbf T_{h}^*=\mathbf T_{\overline{h}}$ and $\mathbf T_{h\partial_j}=\mathbf T_{h}\mathbf T_{\partial_j}$, we get
      \begin{align*}
        \mathbf T_{\mathbf{P}_w} = \sum_{j=1}^n \big(\mathbf T_{\partial_j(rw)/w} \big)^*\, \mathbf T_{\partial_j} = - \sum_{j=1}^n \big( \mathbf T_{r\partial_j} \big)^* \, \mathbf T_{\partial_j} = -\sum_{j=1}^n \mathbf T_{\partial_j}^* \, \mathbf T_{r} \, \mathbf T_{\partial_j}.
      \end{align*}
      So $\mathbf T_{\mathbf{P}_w}$ is not only selfadjoint but even positive on $\bergonw{w}{\Omega}$.
      \medskip

      Now in the context of the unweighted Bergman space, we can even compute the symbol of the operator $\gamma \,\mathbf T_{\mathbf P_{w=1}} \,K$. Indeed, using Stokes' formula, we get for $f,\,g \in A^2(\Omega)$,

      \begin{align*}
        \langle \mathbf T_{\partial_j}\,f,\,g \rangle_{A^2(\Omega)}&=\langle \partial_jf,\,g \rangle_{A^2(\Omega)}=\int_\Omega d\mu\,(\partial_jf)\,\bar g =-\int_{\partial\Omega}d\sigma\, f\,\overline g \,\tfrac{\partial_j r}{2\|\partial r\|}-\int_\Omega d\mu\,f\,\partial_j(\overline g)\\
        &=-\int_{\partial\Omega} d\sigma\,f\,\overline g\, \tfrac{\partial_j r}{2\|\partial r\|}\,.
      \end{align*}
      It follows that, on $W_{hol}^{-1/2}(\partial \Omega)$, $K^* \mathbf T_{\partial_j} K = -T_{\partial_j r/2\|\partial r\|}$ and
      \begin{equation}
        \label{gTK}
        \begin{aligned}
          &\gamma \,\mathbf T_{\partial_j} \,K=-\Lambda^{-1} \,T_{\partial_j r/2\|\partial r\|}\,,\\
          & \mathbf T_{\partial_j} = -V\, T_\Lambda^{-1/2} T_{\partial_j r/2\|\partial r\|} T_\Lambda^{-1/2}\,V^*.
        \end{aligned}
      \end{equation}
      The use of \eqref{princ}, \eqref{princP} and \eqref{upsP} gives
      \begin{align}
        \label{gammasigmaTPnK}
        \sigma(\gamma \mathbf T_{\mathbf P_{w=1}} K)(x',\xi')=-2\,\norm{\eta_{x'}}\norm{\xi'}, \quad (x',\xi')\in\Sigma.
      \end{align}
      In particular $-V^* \,\mathbf T_{\mathbf P_{w=1}}V$ is an elliptic GTO.
      }
    \end{remark}

    \begin{remark}
      {\rm
      The hypothesis $\mathbf{P}=\mathbf{P}^*$ in \lemref{TP=TP*} is quite strong because the equality $\mathbf T^*_{\mathbf P}=\mathbf T_{\mathbf P^*}$ is not necessarily true: when $w=1$, one deduces from \eqref{gTK} that
      $$
        \mathbf T^*_{\partial_j} = -V\, T_\Lambda^{-1/2} T_{\overline{\partial_j r}/2\|\partial r\|} T_\Lambda^{-1/2}\,V^* \neq 0=\mathbf T_{-\overline{\partial}_j}
      $$
      while $-\overline{\partial}_j$ is the formal adjoint of $\partial_j$ on the domain of smooth compactly supported functions on $\Omega$. Of course, any selfadjoint extension of a differential operator $\mathbf P$ needs to take care of the boundary conditions on $\partial \Omega$.
      }
    \end{remark}

    We now establish a result on unitaries $\mathbf{U}$ on $A^2_w(\Omega)$ such that $V_w^*\, \mathbf{U} \, V_w$ is a GTO, result which will be used later on in Section \ref{Spectral triples}.

    \begin{lemma}
      Let $g\in S^0(\R)$, i.e.
      $$
        \forall j\in\mathbb N \, \text{ there exists }c_j(g)>0 \text{ with } \,\,
        |\partial^j_\xi g(\xi)| \le c_j(g)(1+\norm{\xi})^{-j},\,\,\forall \xi\in\R,
      $$
      and let $A$ be an elliptic selfadjoint pseudodifferential operator of order 1 on a compact manifold $M$. Then the operator $\exp(ig(A))$ is a pseudodifferential operator of order 0 on $M$.
    \end{lemma}

    \begin{proof}
      The Fa\`a di Bruno's formula states that for a couple $f,\,g$ of smooth functions on $\R$,
      \begin{align*}
        \partial^n(f \circ g)(\xi) = \sum_{(k_1,\dots,k_n) \in K_n} \tfrac{n!}{k_1!\dots k_n!}\, (\partial^{\sum_j k_j} f)\big(g(\xi)\big) \, \prod_{j=1}^n (\tfrac{1}{j!}\,\partial^j g)^{k_j}(\xi) \, ,\quad \forall \,n\in\N,
      \end{align*}
      where $K_n = \{(k_1,\dots,k_n) \in \N^n, k_1 + 2k_2 + \dots + nk_n = n\}$.

      Thus, if $f : \xi\in \R \mapsto \exp(i\,\xi)$ and $g\in S^0(\R)$, we get
      \begin{align*}
        \abs{\partial^n (f\circ g) (\xi)} \leq \sum_{(k_1,\dots,k_n) \in K_n} \tfrac{n!}{k_1!\dots k_n!}\, \big[ \prod_{j=1}^n (\tfrac{1}{j!}c_j(g))^{k_j} \big] \, (1+\norm{\xi})^{-n}.
      \end{align*}
      So there are constants $c_n(f\circ g) >0$ such that
      \begin{align*}
          \abs{\partial^n(f \circ g)(\xi)} \leq c_n(f\circ g) (1+\norm{\xi})^{-n},
        \tag*{\qed}
      \end{align*}
      \hideqed
      and we have shown that the function $\xi\in \R\mapsto \exp\big(ig(\xi)\big)$ is in $S^0(\R)$. \\
      Applying \cite[Theorem 1]{Strichartz} (or \cite[Theorem 1.2]{TaylorFIO}), we conclude that $\exp(ig(A))$ is a pseudodifferential operator of order 0 on $M$.
    \end{proof}

    \begin{corollary}
      \label{Unitary-Cor}
      Let $\varphi$ be any function in $S^0(\R)$ $($for instance, $\varphi(\xi)\vc(1+\xi)^2(1+\xi^2)^{-1})$ and choose a pseudodifferential operator $A$ on $\partial\Omega$ of order $1$ such that $A$ commutes with $\Pi$ and $T_A=(T_{\Lambda_w})^{-1/(m_w+1)}$.
      Denote $V_\varphi\vc\exp\big(i\varphi(A)\big)$. \\
      Then $T_{V_\varphi}$ is a unitary generalized Toeplitz operator on $H^2(\partial\Omega)$.
    \end{corollary}

    \begin{proof}
      By previous lemma, $V_\varphi\vc\exp\big(i\varphi(A)\big)$ is a unitary classical pseudodifferential operator on $M=\partial \Omega$ which commutes with $\Pi$, thus $T_{V_\varphi}=\exp\big[i\varphi\big(T_{\Lambda_w}^{-1/(m_w+1)}\big)\big]$ is a unitary generalized Toeplitz operator.
    \end{proof}

    \begin{remark}
      \label{Unitary-Rem}
      {\rm
      In \secref{secBergm}, we will need to find unitary operators on $\bergonw{w}{\Omega}$ of the form $V_w T V_w^*$, with $T$ a GTO, to deduce non-positive Dirac-like operators from positive ones. So we can take $\mathbf{U} \vc V_wT_{V_\varphi}V_w^*$, with $T_{V_\varphi}$ defined in \corref{Unitary-Cor}.\\
      Another class of unitary operators on $\bergonw{w}{\Omega}$ can be obtained as follows. Take any GTO $T_P$ which is invertible and not a constant multiple of a positive operator. (For instance, $T_P=T_f$ with $f$ a nonconstant zero-free holomorphic function: the zero-free condition ensures $T_f=M_f$ is invertible, while $T_f=cA$ with $A>0$ would mean that multiplication by the nonconstant holomorphic function $f/c$ is a positive operator, which is quickly seen to lead to contradiction.) We know there exists another pseudodifferential operator $Q$ such that $\Pi \,Q=Q\,\Pi$ and $T_P=T_Q$. Hence also $T^*_P \,T_P=T^*_Q \,T_Q=T_{Q^*Q}=T_{|Q|}^2$, implying that $U\vc T_P\,T_{|Q|}^{-1}$ is a unitary generalized Toeplitz operator. From \propref{gammaTK}, the operator $ \mathbf{U} \vc V_w \, U \, V_w^*$ is unitary from $\bergonw{w}{\Omega}$ onto $\bergonw{w}{\Omega}$. Furthermore, $\mathbf U$ is not a multiple of the identity; for, if it were, then so would be $U$, hence $T_P=UT_{|Q|}$ would be a constant 
multiple of the positive operator $T_{|Q|}$, contrary to the hypothesis.
      }
    \end{remark}

%%%%%%%%%%%%%%%%%%%%%%%%%%%%%%%%%%%%%%%%%%%
\section{Fock space and Heisenberg algebra}
%%%%%%%%%%%%%%%%%%%%%%%%%%%%%%%%%%%%%%%%%%%
  \label{FSandHA}

  %%%%%%%%%%%%%%%%%%%%%%%%%%%%%%%%%%%%%%%%%%%%%%%%%%%%%
  \subsection{Fock space and Segal--Bargmann transform}
  %%%%%%%%%%%%%%%%%%%%%%%%%%%%%%%%%%%%%%%%%%%%%%%%%%%%%

    If $(x,y)\in\R^{2n}$, we denote $xy \vc \sum_{k=1}^n x_k y_k$. For $(z,z')\in\C^{2n}$, we use the usual scalar product $\scalp{z}{z'} \vc \sum_{k=1}^n z_k \bar{z}'_k$. Also, for $(x,z)\in \R^n\times\C^n$, we denote $xz$ the complex number $\sum_{k=1}^n x_k z_k$ (so $x^2=\sum_{k=1}^n x_k^2$ and $z^2=\sum_{k=1}^n z_k^2$), and $\abs{z}^2 \vc \sum_{k=1}^n \abs{z_k}^2$.

    We recall now some known results about the Fock space:
    \begin{definition}
      For $\cstFock > 0$, the Fock space is $$\mathscr{F}_\cstFock \vc L^2_{hol}(\C^n, dm_{\cstFock}(z)),\text{ where }dm_{\cstFock}(z):=(\pi\cstFock)^{-n}\,e^{-\abs{z}^2/\cstFock}\,d\mu(z).$$
    \end{definition}
    This Fock space is denoted $\mathscr{F}$ in \cite[Chap. 1.7]{Howe1980} when $\cstFock = \pi^{-1}$. Here, the measure $dm_\cstFock$ has been chosen such that $m_\cstFock(1)=1$. \\
    The family of functions $\{u_\alpha\}_{\alpha\in\N^n}$, where
    \begin{align}\label{ualpha}
      u_\alpha(z) = a_\alpha\, z^\alpha \vc (\cstFock^{\abs{\alpha}}\alpha!)^{-1/2} \, z^{\alpha},
    \end{align}
    forms an orthonormal basis of $\mathscr{F}_\cstFock$.

%     \begin{proof}
%       We consider the case $n=1$ for simplicity. For $(j,k)\in \N^{2}$, we have
%       \begin{align*}
%         \scalp{z^j}{z^k}_{F_\cstFock} = & \, \tfrac{1}{\pi\cstFock} \int_{z\in\C} z^j \bar{z}^k \,e^{-\tfrac{\abs{z}^2}{\cstFock}} \,dv(z) = \tfrac{1}{\pi\cstFock} \int_{r=0}^{\infty}\int_{\theta=0}^{2\pi} r^{j+k+1}\, e^{i\theta(j-k)} e^{-\tfrac{r^2}{\cstFock}} \,dr \,d\theta.
%       \end{align*}
%       If $j \neq k$, this quantity is zero, while for $j=k$ instead:
%       \begin{align*}
%         \norm{z}^2_{F_\cstFock} = & \, \tfrac{2\pi}{\pi\cstFock} \int_{r=0}^\infty r^{2j+1} \,e^{-\tfrac{r^2}{\cstFock}}\, dr = \tfrac{2}{\cstFock} \tfrac{\cstFock^{1+j} j!}{2} = \cstFock^j j!.
%       \end{align*}
%       So for arbitrary positive integer $n$ and any multiindex $\alpha$, $\norm{z^\alpha}_{F_\cstFock}^2 = \cstFock^{\abs{\alpha}}\alpha!$.\\
%       The family $\{u_\alpha\}_{\alpha \in \N^n}$ is dense in $F_\cstFock$:
%    %we show that we have, for any $\alpha\in\N^n$ and any $f\in F_\cstFock$, $\scalp{u_\alpha}{f}_{F_\cstFock} = 0$ implies $f=0$:
%       let $f\in F_\cstFock \subset L^2_{hol}(\C^n)$ be orthogonal to all $\{u_\alpha\}_\alpha$. So, for any $z\in\C^n$, $f(z)=\sum_{\beta\in\N^n} f_\beta \,z^\beta$, and
%       \begin{align*}
%        0= \scalp{u_\alpha}{f}_{F_\cstFock} = & \, (\cstFock^{\abs{\alpha}}\alpha!)^{-1/2}(\pi\cstFock)^{-n}\sum_{\beta\in\N^n}f_\beta \int_{z\in\C^n}z^\alpha \bar{z}^\beta\,dv(z) = (\cstFock^{\abs{\alpha}}\alpha!)^{1/2}(\pi\cstFock)^{-n} f_\alpha,
%       \end{align*}
%   thus $f_\alpha = 0$ and $f=0$.
%     \end{proof}

    \begin{definition}\label{SegBarg}
      The Segal--Bargmann transform $W_\cstFock : L^2(\R^n,dx) \to \mathscr{F}_\cstFock$ is
      \begin{align*}
        (W_\cstFock\, f)(z) \vc (\pi\cstFock)^{-n/4}\int_{x\in \R^n} e^{(-z^2 + 2\sqrt{2}\,xz - x^2)/2\cstFock} \, f(x)\, dx.
      \end{align*}
    \end{definition}

    \begin{prop}
      \label{Wtunitary}
      The Segal--Bargmann transform $W_\cstFock$ is a unitary map from $L^2(\R^n)$ to $\mathscr{F}_\cstFock$.\\
      Moreover, for any $f\in \mathscr{F}_\cstFock$,
      \begin{align*}
        (W_\cstFock^{-1}f)(x) = (\pi\cstFock)^{-5n/4}\int_{z\in\C^n} e^{(-\bar{z}^2 + 2\sqrt{2}\,x\bar{z} - x^2)/2\cstFock} \, f(z) \, e^{-\abs{z}^2/\cstFock} \, d\mu(z) .
      \end{align*}
    \end{prop}
    A proof is given in the Appendix.

    The map $W_\cstFock$ plays an important role since it links unitarily functions of real variables to holomorphic functions of complex variables.

  %%%%%%%%%%%%%%%%%%%%%%%%%%%%%%%%%%%%%%%%%%%%%%%%%%%%%
  \subsection{The Heisenberg group and its Lie algebra}
  %%%%%%%%%%%%%%%%%%%%%%%%%%%%%%%%%%%%%%%%%%%%%%%%%%%%%

    \begin{definition}
      The Heisenberg group $\mathbb{H}^n$ is the set $\R^n\times\R^n\times \R$ endowed with product:
      \begin{align*}
        (q,p,s)\,(q',p',s') = \big(q+q',p+p', s+s'+\tfrac{1}{2}(qp'-pq')\big),
      \end{align*}
      where $(q,p)\in \R^{2n}$ and $s\in\R$.
    \end{definition}
    The unit element is $(0,0,0)$ and $(q,p,s)^{-1}=(-q,-p,-s)$.
    \\
    Let us define for $j\in\set{1,\cdots,n}$ the generators $Q_j$, $P_j$ and $T$ of $\mathbb{H}^n$ as:
    \begin{align*}
      & \exp(Q_j) \vc  \, (1_j,(0,\dots,0),0),\\
      & \exp(P_j) \vc  \, \big((0,\dots,0),1_j,0\big),\\
      & \exp(T) \vc  \, \big((0,\dots,0),(0,\dots,0),1\big),
    \end{align*}
    where $1_k$ denotes the multiindex being zero everywhere and 1 at the $k^{th}$ position.
    \\
    These generators form a basis of the Lie algebra $\mathfrak{h}^n$ of $\mathbb{H}^n$. The only non null commutation relations are $\commut{Q_j}{P_k}= \delta_{j,k}\,T$.
    We also define the two elements of $\mathfrak{h}^n$
    \begin{align*}
      a_j \vc \tfrac{1}{\sqrt{2}}(Q_j+iP_j)\,\text{ and }\, a_j^+ \vc \tfrac{1}{\sqrt{2}}(Q_j-iP_j),
    \end{align*}
    which verify $\commut{a_j}{a_k^+} = -\tfrac{i}{2}(\commut{Q_j}{P_k} + \commut{Q_k}{P_j}) = -i\delta_{j,k}\,T$.\\
    We will also use the element $N$ of the universal enveloping algebra $\envalg{\mathfrak{h}^n}$ of $\mathfrak{h}^n$:
    \begin{align*}
      N \vc \tfrac{1}{2}\sum_{j=1}^n a_j^+a_j + a_ja_j^+.
    \end{align*}

  %%%%%%%%%%%%%%%%%%%%%%%%%%%%%%%%%%%%%%%%%%%%%%%%%%%%%
  \subsection{\texorpdfstring{Representations of $\mathfrak{h}^n$}{Representation of H}}
  %%%%%%%%%%%%%%%%%%%%%%%%%%%%%%%%%%%%%%%%%%%%%%%%%%%%%

    Different representations of $\mathfrak{h}^n$ on the functions spaces we consider here have been studied in \cite{Howe1980} and later by \cite{Taylor1984} with different approaches.

    %%%%%%%%%%%%%%%%%%%%%%%%%%%%%%%%%%%%%%%%%%
    \subsubsection{Schrödinger representation}
    %%%%%%%%%%%%%%%%%%%%%%%%%%%%%%%%%%%%%%%%%%
      \label{section331}

      \begin{definition}
        The Schrödinger representation $\rho_t$ of $\mathfrak{h}^n$ on $L^2(\R^n,dx)$ is defined by
        \begin{align*}
          & \rho_{\cstSch}(Q_j)f(x) \vc \, x_j\,f(x),\quad \rho_{\cstSch}(P_j)f(x) \vc -i\cstSch\,\partial_{x_j}f(x), \quad \rho_{\cstSch}(T)f(x) \vc \, i\cstSch\, f(x),
        \end{align*}
        where $\cstSch$ is a strictly positive parameter.
      \end{definition}
      With this representation, the canonical commutation relations from the quantum mechanics can be recovered by setting $\cstSch=\hbar$, the Planck constant, $\widehat{x_j}=\rho(Q_j)$ and $\widehat{p_j}=\rho(P_j)$, respectively the position and momentum operators:
      $\commut{\widehat{x_j}}{\widehat{p_j}} = -i\hbar\, \partial_{x_j} + i\hbar \,\bbbone + i\hbar \,\partial_{x_j} = i\hbar\, \bbbone = \rho(\commut{Q_j}{P_j}).$\\
      We have also
      \begin{align*}
        \rho_{t'}(a_j) = \tfrac{1}{\sqrt{2}}(x_j + \cstSch\,\partial_{x_j}), \quad \rho_{t'}(a_j^+) = \tfrac{1}{\sqrt{2}}(x_j - \cstSch \,\partial_{x_j}), \quad \rho_{t'}(N) = \tfrac{1}{4}\sum_{j=1}^n (x_j^2 - \cstSch^2\,\partial_{x_j}^2).
      \end{align*}

    %%%%%%%%%%%%%%%%%%%%%%%%%%%%%%%%%%%
    \subsubsection{Fock representation}
    %%%%%%%%%%%%%%%%%%%%%%%%%%%%%%%%%%%

      From the Schrödinger representation $\rho_{\cstSch}$, we use the unitary map $W_\cstFock$ to get a unitary representation of $\mathfrak{h}^n$ on the Fock space $\mathscr{F}_\cstFock$ by  choosing $\cstSch = \cstFock$.
      \begin{definition}
        The Fock representation $\nu_\cstFock$ of an element $h\in\mathfrak{h}^n$ on $\mathscr{F}_\cstFock$ is
        \begin{align*}
          \nu_\cstFock(h) \vc W_\cstFock \, \rho_\cstFock(h) \, W_\cstFock^{-1}.
        \end{align*}
      \end{definition}

      \begin{prop}
        \label{nutaction}
        The explicit actions on the basis vector of $\mathscr{F}_\cstFock$ are given by
        \begin{align*}
          \begin{aligned}
            & \nu_\cstFock (Q_j) \, u_\alpha  =  (\tfrac{\cstFock}{2})^{1/2}\,(\sqrt{\alpha_j}\,u_{\alpha-1_j} + \sqrt{\alpha_j+1}\,u_{\alpha+1_j}), \\
            & \nu_\cstFock (P_j) \, u_\alpha  = -i(\tfrac{\cstFock}{2})^{1/2}\,(\sqrt{\alpha_j} \, u_{\alpha -1_j} - \sqrt{\alpha_j +1} \, u_{\alpha +1_j}), \\
            & \nu_\cstFock (T) \, u_\alpha  = i\cstFock \, u_\alpha, \\
            & \nu_\cstFock (a_j) \, u_{\alpha}  = \cstFock^{1/2}\,\sqrt{\alpha_j} \, u_{\alpha-1_j}, \quad \nu_\cstFock (a_j^+) \,  u_{\alpha}  = \cstFock^{1/2}\,\sqrt{\alpha_j+1} \, u_{\alpha+1_j}, \quad \nu_\cstFock (N) \, u_\alpha  = \cstFock \,(\abs{\alpha} + \tfrac{n}{2}) \, u_\alpha.
          \end{aligned}
        \end{align*}

      \end{prop}

      \begin{proof}
        Differentiating in Definition \ref{SegBarg}, we get (see Appendix)
        $$
          (\partial_{z_j} W_t f )(z) = \big(W_t (-\tfrac{z_j}t+\tfrac{\sqrt2 x_j}t)f \big)(z), \text{ so } \tfrac{z_j+t\partial_{z_j}}t \,W_t = W_t \,\tfrac{\sqrt2 x_j}t, \text{ or } W_t\, \tfrac{\sqrt2 x_j}t\, W_t^{-1} = \tfrac{z_j}t+\partial_{z_j}.
        $$
        Similarly, integrating by parts in Definition \ref{SegBarg} we get $\big(W_t(\partial_{x_j} f)\big) (z) = \big(W_t (\tfrac{x_j}t-\tfrac{\sqrt2 z_j}t)f\big) (z)$, so $
        W_t \Big(\tfrac{x_j}t-\partial_{x_j}\Big) = \tfrac{\sqrt2 z_j}t W_t,\text{ or }W_t \,\Big(\tfrac{x_j}t-\partial_{x_j}\Big) \,W_t^{-1} = \tfrac{\sqrt2 }t \,z_j$. Consequently,
        $$
          W_t \,x_j \,W_t^{-1}= \tfrac{1}{\sqrt2}(z_j+t\,\partial_{z_j}),  \qquad
          W_t \,\partial_{x_j} \,W_t^{-1} = \tfrac{1}{\sqrt2}(-z_j+t\,\partial_{z_j}).
        $$
        Using the formulas in \S\ref{section331}, we get the result.
      \end{proof}

    %%%%%%%%%%%%%%%%%%%%%%%%%%%%%%%%%%%%%%
    \subsubsection{Bergman representation}
    %%%%%%%%%%%%%%%%%%%%%%%%%%%%%%%%%%%%%%

      The Fock and Bergman spaces are linked through a simple change of basis via the unitary operator $\mathscr{V}_\cstFock^w: \mathscr{F}_\cstFock \to \bergonw{w}{\Omega}$:
      \begin{align}
        \label{vprimealpha}
        \mathscr{V}_\cstFock ^w \,(u_\alpha) \vc v_\alpha^w\, ,
      \end{align}
      where $u_\alpha$ as in \eqref{ualpha} and $v_\alpha^w$ any orthonormal basis of $\bergonw{w}{\Omega}$.

      We will also need the map
      \begin{align} \label{Ufock}
        U^w_\cstFock \vc \mathscr{V}^w_\cstFock\, W_\cstFock,
      \end{align}
      which is by construction, a unitary operator from $L^2(\R^n,dx)$ to $\bergonw{w}{\Omega}$.

      For the peculiar case $\Omega=\mathbb{B}^n$, $\partial\Omega=S^{2n-1}$ and $w=r^0=1$, we can get explicit formulas. Using the basis \eqref{valpha} of $\berg$, we denote the above change of basis by $\mathscr{V}_\cstFock : \mathscr{F}_\cstFock \to \berg$:
      \begin{align}
        \label{Vfock}
        \mathscr{V}_\cstFock\,(u_\alpha) \vc v_\alpha\,.
      \end{align}
      Here again, we fix the constant $\cstSch$ from the Schrödinger representation $\rho_\cstSch$ with $\cstSch = \cstFock$ to define the Bergman representation.
      \begin{definition}
        The Bergman representation $\tau^w_\cstFock$ of an element $h\in\mathfrak{h}^n$ on $\bergonw{w}{\Omega}$ is given by
        \begin{align*}
          \tau^w_\cstFock (h) \vc \mathscr{V}^w_\cstFock\, \nu_\cstFock(h) \,{\mathscr{V}^w_\cstFock}^{-1} = U^w_\cstFock\, \rho_t(h) \,{U^w_\cstFock}^{-1}.
        \end{align*}
        We denote by $\tau_\cstFock$ this representation in the case $\Omega=\mathbb{B}^n$, $\partial\Omega=S^{2n-1}$ and $w=r^0=1$.
      \end{definition}
      We get directly the following results, since the Bergman representation differs from the Fock representation by a simple change of basis:

      \begin{prop}
        \label{etaRep}
        The representation $\tau^w_\cstFock$ has the following properties:
        \begin{align*}
          \begin{aligned}
            & \tau^w_\cstFock (Q_j) \, v_\alpha^w  = (\tfrac{\cstFock}{2})^{1/2}\,\big(\sqrt{\alpha_j} \, v_{\alpha-1_j}^w + \sqrt{\alpha_j+1} \, v_{\alpha+1_j}^w\big), \\
            & \tau^w_\cstFock (P_j) \, v_\alpha^w  = -i(\tfrac{\cstFock}{2})^{1/2}\,\big(\sqrt{\alpha_j} \, v_{\alpha -1_j}^w - \sqrt{\alpha_j +1} \, v_{\alpha +1_j}^w\big), \\
            & \tau^w_\cstFock (T) \, v_\alpha^w  = i\cstFock \, v_\alpha^w, \\
            & \tau^w_\cstFock (a_j) \, v_{\alpha}^w  = \cstFock^{1/2}\sqrt{\alpha_j} \, v_{\alpha-1_j}^w, \quad \tau^w_\cstFock (a_j^+) \, v_{\alpha}^w  = \cstFock^{1/2}\sqrt{\alpha_j+1} \, v_{\alpha+1_j}^w, \quad \tau^w_\cstFock (N) \, v_\alpha^w  = \cstFock (\abs{\alpha} + \tfrac{n}{2}) \, v_\alpha^w.
          \end{aligned}
        \end{align*}
      \end{prop}

  %%%%%%%%%%%%%%%%%%%%%%%%%%%%%%%%%%%%%%%%%%%%%%%%%%
  \subsection{A diagram as a summary}
  %%%%%%%%%%%%%%%%%%%%%%%%%%%%%%%%%%%%%%%%%%%%%%%%%%
    \label{secdiag}

    The figure below shows a diagram with all the representations and maps involved here: the respectively dotted, dashed, simple, thick arrows refer to injections, surjections, isomorphisms, isometries. Double arrows indicate the action of the Lie algebra $\mathfrak{h}^n$ on the Hilbert spaces.

    We must highlight the fact that only composition of maps which does not involve dotted or dashed lines (projectors) are commutative. Indeed, for instance $U_w\boldsymbol{\Pi}_w \neq \Pi V_w$.\\
    Note that $K_w$ maps $L^2(\partial \Omega)$ into $L^2_{harm}(\Omega,w)$ with dense range and that $\gamma \, \mathbf{\Pi}_w \,K_w= T_{\Lambda_w}^{-1}\,\Pi\,\Lambda_w$ on $L^2(\partial\Omega)$ (see \cite[Proposition 8]{E2010}), so $\gamma \, \mathbf{\Pi}_w \,K_w \Pi=\Pi$.

    \begin{figure}[!h]
      \centering
      \begin{tikzpicture}
        \node (A)  {$L^2_{hol}(\Omega,w)\cv \bergonw{w}{\Omega}=W_{hol}^{-m_w/2}(\Omega)$};
        \node (AA) [node distance=0.25cm, right of=A,above of=A]{ };
        \node (h) [node distance=2cm, below of=A] {$\mathfrak{h}^n$};
        \node (L) [node distance=2cm, right of=h, below of=h] {$L^2(\R^n)$};
        \node (F) [node distance=2cm, left of=h, below of=h] {$\mathscr{F}_\cstFock$};
        \node (H) [node distance =6cm,  left of=A] {$L^2_{hol}(\partial\Omega)\cv H^2(\partial\Omega)$};
        \node (W) [node distance =5cm, left of=H] {$W_{hol}^{-(m_w+1)/2}(\partial\Omega) $};
        \node (gammaA) [node distance =1.3cm, below of =W] {$\gamma_w\big(A^2_w(\Omega)\big)$};
        \node (L2) [node distance =4cm, above of =H] {$L^2(\partial\Omega)$};
        \node (L2harm) [node distance =2cm, above of =A] {$L^2_{harm}(\Omega,w)$};
        \node (L2w) [node distance =2cm, above of =L2harm] {$L^2(\Omega,w)$};
        \path [-latex] (h) edge [double] node[left]  {$\rho\,$ } (L) ;
        \path [-latex] (h) edge [double] node[right]  {$\,\nu_\cstFock$} (F) ;
        \path [-latex] (h) edge [double] node[right]  {$\tau^w_t$} (A) ;
        \path [->,>=stealth,ultra thick] (L) edge node[below,pos=0.4]  {$W_\cstFock$} (F) ;
        \path [->,>=stealth,ultra thick] (L) edge[bend right,looseness=0.7] node[below,right,pos=0.6]  {$U^w_\cstFock$} (A) ;
        \path [->,>=stealth,ultra thick] (F) edge[bend left,looseness=0.6] node[below,right,pos=0.6]  {$\mathscr{V}^w_\cstFock\,$} (A) ;
        \path [->,>=stealth,ultra thick] (H) edge node[above=9, right,pos=0.15]  {$V_w$} (A) ;
        \path [->,>=stealth] (W) edge[bend right,looseness=0.3] node[below=8,left,pos=0.5]  {$K_w$} (A) ;
          \path [->,>=stealth] (A) edge[bend right,looseness=0.3] node[above=8,left,pos=0.45]  {$\gamma_w$} (W) ;
        \path [->,>=stealth] (W) edge node[above,pos=0.4]  {$T_{\Lambda_w}^{1/2}$} (H) ;
        \path [->, dashed,->,>=stealth] (L2) edge node[left]  {$\Pi$} (H) ;
        \path [->,>=stealth,ultra thick] (L2) edge node[above]  {$U_w$} (L2harm) ;
        \path [->, dashed,>=stealth] (L2harm) edge node[right]  {$\mathbf{\Pi}_w$} (A) ;
        \path [->, dashed,>=stealth] (L2w) edge node[right]  {$\mathbf{\Pi}_{w,harm}$} (L2harm) ;
        \path [->, dotted,>=stealth,ultra thick] (L2) edge[bend right,looseness=0.5]node[below=5,left,pos=0.6]  {$K_w$} (L2harm) ;
        \path [->, dotted,>=stealth,ultra thick] (L2harm) edge[bend right,looseness=0.5] node[above=12,left,pos=0.4]  {$\gamma_w$} (L2) ;
        \node[node distance=0.70cm, below of=W, rotate=90] {=} ;
      \end{tikzpicture}
    \end{figure}

    For the case $\Omega=\mathbb{B}^n$, $w=r^0=1$, replace $W_{hol}^{-(m_w+1)/2}(\partial\Omega)$, $\mathscr{V}^w_\cstFock$, $U^w_\cstFock$ and $V_w$ of \eqref{V_w} respectively by $W_{hol}^{-1/2}(\partial\Omega)$, $\mathscr{V}_\cstFock$ as in \eqref{Vfock}, $U_\cstFock$ as in \eqref{Ufock} and $V$.

%%%%%%%%%%%%%%%%%%%%%%%%%%%%%%%%%%%%%%%%%%%%%%%%%%%%
\section{Pseudodifferential and Toeplitz operators }
%%%%%%%%%%%%%%%%%%%%%%%%%%%%%%%%%%%%%%%%%%%%%%%%%%%%

  The relation between pseudodifferential operators and Toeplitz operators has been studied in \cite{Howe1980}, \cite{BMG1981} and \cite{Guillemin1984}. The authors show that these operators enjoy a similar symbolic calculus. Moreover, there exists an isomorphism between pseudodifferential operators of order $k$ on $L^2(\R^n)$ and Toeplitz operators of order $k/2$ on the Bergman space $\berg$. This isomorphism is nothing else than the conjugation by the unitary $U_\cstFock$ defined above (see \cite[Appendix B]{Taylor1984} for a detailed proof).

  %%%%%%%%%%%%%%%%%%%%%%%%%%%%%%%%%%%%%%%%%%%%%%%%%%%%%%
  \subsection{\texorpdfstring{Toeplitz operators on $\berg$ as representation of elements in $\envalg{\mathfrak{h}^n}$}{Toeplitz operator ...}}
  %%%%%%%%%%%%%%%%%%%%%%%%%%%%%%%%%%%%%%%%%%%%%%%%%%%%%%

    We express here the Toeplitz operators $\mathbf T_{z^\alpha}$ and $\mathbf{T}_{\partial^\alpha}$ acting on the Bergman space $\berg$ as representations of elements in the enveloping algebra $\envalg{\mathfrak{h}^n}$. We follow \cite[Chap. 4.2]{Howe1980} for the first result.

    \begin{prop}\label{Tasheisenberg}
      For $\alpha\in\N^n$, let $a^\alpha \vc \prod_{j=1}^n a_j^{\alpha_j}$, $(a^+)^\alpha \vc \prod_{j=1}^n (a_j^+)^{\alpha_j}$ (the order is not important since $\commut{a_j}{a_k} = \commut{a_j^+}{a_k^+} = 0$), and $g_\alpha^{\pm} \vc \prod_{k=1}^{\abs{\alpha}} \big( N -i(\tfrac{n}{2} \pm k)T\big)$ be elements of $\envalg{\mathfrak{h}^n}$.
      The Toeplitz operators $\mathbf{T}_{z^\alpha}$ and $\mathbf{T}_{\partial^\alpha}$ from $\berg$ to $\berg$ can be written as
      \begin{align*}
        \mathbf{T}_{z^\alpha} = \tau_\cstFock \big( (a^+)^\alpha \big) \, \big( \tau_\cstFock (g_\alpha^+) \big)^{-1/2} \quad \text{ and } \quad \mathbf{T}_{\partial^\alpha} = \tau_\cstFock \big( a^\alpha \big) \, \big( \tau_\cstFock (g_\alpha^-) \big)^{-1/2}.
      \end{align*}
    \end{prop}
    Notice that the $t$-dependence in the respective right hand sides compensate, so these sides do not depend on $t$.

    \begin{proof}
      Since $z\mapsto z^\alpha$ is holomorphic, $\mathbf{T}_{z^\alpha} = \mathbf{M}_{z^\alpha}$, and we have also $\mathbf{T}_{\partial^\alpha} = \mathbf{\Pi}_w \partial^\alpha = \partial^\alpha$ on $\berg$. Applied on the basis, they yield
      \begin{equation} \label{extra}
      \begin{aligned}
        & \mathbf{T}_{z^\alpha}\, v_\beta = z^\alpha \,v_{\beta} = \, \tfrac{b_\alpha}{b_{\beta+\alpha}} \,v_{\beta+\alpha} = \big(\tfrac{(\beta+\alpha)!}{\beta!}\, \tfrac{(\abs{\beta}+n)!}{(\abs{\beta}+\abs{\alpha} + n)!}\big)^{1/2}\, v_{\beta+\alpha} \, , \quad \forall \alpha, \beta \, , \\
        & \mathbf{T}_{\partial^\alpha}\, v_\beta = \partial^\alpha \,v_{\beta} = \, \tfrac{\beta!}{(\beta-\alpha)!} \tfrac{b_\alpha}{b_{\beta-\alpha}} \,v_{\beta-\alpha} = \big(\tfrac{\beta!}{(\beta-\alpha)!}\,\tfrac{(\abs{\beta}+n)!}{(\abs{\beta}-\abs{\alpha} + n)!}\big)^{1/2}\, v_{\beta-\alpha} \, , \quad \beta \geq \alpha \, ,
      \end{aligned} \end{equation}
      where for the second equality, the case $\alpha > \beta$ corresponds to the null operator. So for the rest, we only consider $\mathbf{T}_{\partial^\alpha}$ on the domain Span$\{v_\beta\}_{\beta \geq \alpha}$.
      From \propref{etaRep}, we deduce
      \begin{align*}
        & \tau_\cstFock \big((a^+)^\alpha\big)\, v_\beta = \cstFock^{\abs{\alpha}/2} \big(\tfrac{(\beta+\alpha)!}{\beta!}\big)^{1/2} \, v_{\beta+\alpha}, \quad \forall \alpha, \beta, \\
        & \tau_\cstFock (a^\alpha)\, v_\beta = \cstFock^{\abs{\alpha}/2} \big(\tfrac{\beta!}{(\beta-\alpha)!}\big)^{1/2} \, v_{\beta-\alpha}, \quad \quad \beta \geq \alpha.
      \end{align*}
      Moreover, the representations of the elements $g_\alpha^{\pm}$ act on $\berg$ as
      \begin{align*}
        \tau_t (g_\alpha^\pm) \, v_\beta &=  \, \prod_{k=1}^{\abs{\alpha}}\big( \tau_t(N) -i(\tfrac{n}{2} \pm k)\tau_t(T)\big)\, v_\beta = \cstFock^{\abs{\alpha}} \prod_{k=1}^{\abs{\alpha}}(\abs{\beta} + n \pm k) \, v_\beta= \, \cstFock^{\abs{\alpha}} \,\tfrac{(\abs{\beta} \pm \abs{\alpha} + n)!}{(\abs{\beta} + n)!}\, v_\beta \,.
      \end{align*}
      Thus the operators $\tau_t (g_\alpha^{\pm})$ are invertible on $\berg$ and we get the claimed formulae.
    \end{proof}

    \begin{lemma}
      \label{gammatauPjK}
      The operator
      $$
        \mathbf{R} \vc \sum_{j=1}^n \mathbf T_{z_j} \mathbf{T}_{\partial_j}=\sum_{j=1}^n \mathbf T_{z_j\partial_j}   \quad\text{on $\berg$}
      $$
      is positive and
      \begin{align}
        \tau_t(P_j) = -i (\tfrac t2)^{1/2} \,\big[\mathbf{T}_{\partial_j}\,(\mathbf{R}+n)^{-1/2} - (\mathbf{T}_{\partial_j}\,(\mathbf{R}+n)^{-1/2} )^*\big]  \label{tau(Pj)}
      \end{align}
      (showing again that $\tau_t(P_j)$ is selfadjoint).\\
      The operator $\gamma \tau_t(P_j) K$ is a GTO of order $\tfrac 12$ with principal symbol
      \begin{align}
        \label{sigmatau(Pj)}
        \sigma\big(\gamma \tau_t(P_j) K\big)(x',\xi')=2^{3/4}(\tfrac{t}{2})^{1/2}\big(\tfrac{\norm{\eta_{x'}}}{-R(r)(x')} \big)^{1/2}\,\tfrac{\xi'_j}{\norm{\xi'}^{1/2}} \,, \quad (x',\xi')\in\Sigma.
      \end{align}
    \end{lemma}

    \begin{proof}
      By \eqref{extra},  $\mathbf{R} \,v_\beta = |\beta| \,v_\beta   \text{ for any multiindex }\beta$, so $\mathbf{R}\geq 0$. Since
      \begin{align*}
        &\mathbf T_{\partial_j}\,(\mathbf R+n)^{-1/2} \,v_\alpha=\sqrt{\alpha_j} \,v_{\alpha-1_j}, \quad(\mathbf R+n)^{1/2}\mathbf T_{z_j}\,v_\alpha=(\alpha_j+1)^{1/2}\,v_{\alpha+1_j}\\
        & (\mathbf T_{z_j})^* \,v_\alpha= \big(\tfrac{\alpha_j}{\vert \alpha\vert+n}\big)^{1/2} \,v_{\alpha-1_j}\,,
      \end{align*}
      we get
      \begin{align}
        \label{T*z}
        \mathbf T_{z_j}^*=\,(\mathbf R+n+1)^{-1} \mathbf T_{\partial_j}=\mathbf T_{\partial_j}\,(\mathbf R+n)^{-1}.
      \end{align}
      (These relations could also be deduced from $(T_{z_j}^{\mathscr{F}_\cstFock})^*=t\,T_{\partial_j}^{\mathscr{F}_\cstFock}$ where $T^{\mathscr{F}_\cstFock}$ refers to the Toeplitz operator on the Fock space, and for instance,  $\mathbf T_{\partial_j}  = \sqrt{t} \, \mathscr{V}_\cstFock \, T_{\partial_j}^{\mathscr{F}_\cstFock} \, \mathscr{V}_\cstFock^{-1} \, (\mathbf{R}+n)^{1/2}$.)\\
      By \propref{etaRep}, this yields
      $$
        \tau_t(P_j) \,v_\alpha = -i (\tfrac t2)^{1/2} \Big[ \tfrac{1}{\sqrt{|\alpha|+n}} \,\mathbf{T}_{\partial_j} v_\alpha - \sqrt{|\alpha|+n+1} \,\mathbf T_{z_j} v_\alpha \Big],
      $$
      and $\tau_t(P_j) = -i (\tfrac t2)^{1/2} \,\big[\mathbf{T}_{\partial_j}\,(\mathbf{R}+n)^{-1/2} - \mathbf T_{z_j}(\mathbf{R}+n+1)^{1/2}\big]$. Thus,  \eqref{T*z} implies \eqref{tau(Pj)}.

      As in the proof of \remref{TPselfadjoint}, we get for $f,\,g\in A^2(\Omega)$ (omitting the sum over $j$):
      \begin{align}
        \label{Stokes}
        \scalp{ \,(\mathbf R+n)f}{g}=\int_\Omega d\mu\,\partial_j(z_j f)\,\overline{g}=\int_\Omega d\mu\, \partial_j(z_j f\overline{g})=-\int_{\partial \Omega}d\sigma\, f\,\overline{g} \,\tfrac{z_j\partial_j r}{2\norm{\partial r}}\,.
      \end{align}
      Since $\mathbf R+n\geq 0$, $-\mathbf R (r)$ is a positive function on $\partial\Omega$ (for instance, when $r(z)=1-\norm{z}^2$, $\mathbf R(r)(x')=-\norm{x'}^2$).
      So, on $W_{hol}^{-1/2}(\partial \Omega)$,
      $K^*(\mathbf R+n)\,K=T_{-\mathbf R(r)/2\norm{\partial r}}\cv X$ with $X\geq 0$, implying $\mathbf R+n=V\,T_\Lambda^{-1/2}X T_\Lambda^{-1/2}\,V^*$ and $(\mathbf R+n)^{-1/2}=V\,[T_\Lambda^{-1/2}X T_\Lambda^{-1/2}]^{-1/2}\,V^*$. \\
      Finally $\gamma \,(\mathbf R+n)^{-1/2}\,K=T_\Lambda^{-1/2}\,[T_\Lambda^{-1/2}X T_\Lambda^{-1/2}]^{-1/2}\,T_\Lambda^{1/2}$.
      Thus
      \begin{align*}
        \sigma(\gamma\,(\mathbf R+n)^{-1/2}\,K)=\sigma(\Lambda)^{1/2}\,\sigma(X)^{-1/2}
      \end{align*}
      (so this symbol is positive as it has to be).
      \\
      Using \eqref{tau(Pj)} with $A=\gamma\,\mathbf T_{\partial_j}\,K\,\gamma\,(\mathbf R+n)^{-1/2}\,K$, we obtain
      \begin{align*}
        \sigma\big(\gamma \tau_t(P_j) K\big)&=-i(\tfrac{t}{2})^{1/2}[\sigma(A)-\sigma(A^*)]=2(\tfrac{t}{2})^{1/2} \,\text{Im} (\sigma(A))\\
        &=2(\tfrac{t}{2})^{1/2}\,\sigma(\gamma\,(\mathbf R+n)^{-1/2}\,K)\,\text{Im}( \sigma\big(\mathbf \gamma \,\mathbf T_{\partial_j}\, K)\big).
      \end{align*}
      Since $\gamma\,\mathbf T_{\partial_j}\, K=\Lambda^{-1} \,T_{-\partial_j r/2\norm{\partial r}}$,
      we get $\sigma(\gamma\,\mathbf T_{\partial_j}\, K)= \sigma(\Lambda)^{-1}\,(-\partial_j r/2\norm{\partial r})$. Finally,
      \begin{align*}
        \sigma\big(\gamma \tau_t(P_j) K\big)(x',\xi') &= -2(\tfrac{t}{2})^{1/2} \, \Big(\sigma(\Lambda)^{-1/2}\,\sigma(X)^{-1/2}\,(2\norm{\partial r})^{-1} \,\text{Im} (\partial_j r) \Big) \big(x',\xi'\big)\\
        &= -2(\tfrac{t}{2})^{1/2} \, (2 \norm{\xi'})^{1/2} \, \big(\tfrac{-\mathbf R(r)(x')}{2\norm{\partial r}}\big)^{-1/2} \, (2\norm{\partial r})^{-1} \,\text{Im} (\partial_j r)(x')\\
        &= -2(\tfrac{t}{2})^{1/2} \, \big(\tfrac{\norm{\xi'}}{\norm{\partial r}(-R(r)(x'))} \big)^{1/2}\, \text{Im} (\partial_j r)(x')\\
        &= -2(\tfrac{t}{2})^{1/2} \, \big(\tfrac{\norm{\xi'}}{\sqrt{2}\norm{\eta_{x'}}(-R(r)(x'))} \big)^{1/2}\, \text{Im} (\partial_j r)(x').
      \end{align*}
      Moreover, $-\text{Im}(\partial_j r)(x') = (\eta_{x'})_j = \tfrac{\norm{\eta_x'}}{\norm{\xi'}} \xi'_j$, which yields the result.
    \end{proof}

  %%%%%%%%%%%%%%%%%%%%%%%%%%%%%%%%%%%%%%%%%%%%%%%%%%%%%%%%%
  \subsection{\texorpdfstring{Dirac-like operators on $\bergonw{w}{\Omega}$ and $H^2(\partial\Omega)$}{}}
  %%%%%%%%%%%%%%%%%%%%%%%%%%%%%%%%%%%%%%%%%%%%%%%%%%%%%%%%%

    Our goal is to construct spectral triples $(\A,\H,\DD)$ (see Definition \ref{DefSpTr} below) using the algebra $\A$ of Toeplitz operators acting on Hilbert spaces $\H$ of functions on bounded domains such as Bergman and Hardy spaces. The natural candidates  $\mathcal{D}_{\partial\Omega}$ and $\mathcal{D}_{\Omega}$ defined below are the images of the usual Dirac operator $\dirac \vc -i \sum_{j=1}^n {\text{\footnotesize$\Gamma_j$}}\,\partial_{x_j}$ on $\R^n$ through the maps involved in the diagram of Section \ref{secdiag}. Here the {\text{\footnotesize$\Gamma_j$}} are the usual selfadjoint gamma matrices which represent the $n$-dimensional Clifford algebra $\mathcal C_n \cong \C^{2^{[n/2]}}$. The resulting operators act on $H^2(\partial\Omega)\otimes\mathcal C_n$ or $\bergonw{w}{\Omega}\otimes\mathcal C_n$ respectively, and are even a generalized Toeplitz operator of order 1/2, according to previous result.

    Since the representation of $P_j \in \mathfrak{h}^n$ on $L^2(\R^n)$ is $\rho_t(P_j)=-i\cstFock \,\partial_{x_j}$, we define:
    \begin{align}
      & \mathcal{D}_{\partial\Omega} \vc (\tfrac{\cstFock}{2})^{-1/2}\sum_{j=1}^n {\text{\footnotesize $\Gamma_j$}} \, V_w^* \, \tau^w_\cstFock(P_j) \, V_w\,,  \label{newD2}\\
      & \mathcal{D}_{\Omega} \vc V_w \, \mathcal{D}_{\partial\Omega} \, V_w^* = (\tfrac{\cstFock}{2})^{-1/2}\sum_{j=1}^n {\text{\footnotesize$\Gamma_j$}}\,U_\cstFock^w \, \rho_\cstFock(P_j)\, {U_\cstFock^w}^{*} \label{newD1},
    \end{align}
    acting respectively on $H^2(\partial\Omega)\otimes\mathcal C_n$ and $\bergonw{w}{\Omega}\otimes\mathcal C_n$ . They are selfadjoint and again, do not depend on $t$.

    Another way to construct operators $\DD$ on $\bergonw{w}{\Omega}$ is to consider an operator $\mathbf{T}_{\mathbf{P}}$ of the form \eqref{upsPX}, with $\mathbf{P}$ a selfadjoint differential operator on $L^2(\Omega,w)$ as in \propref{propTPpsdo}.

%%%%%%%%%%%%%%%%%%%%%%%%%%
\section{Spectral triples}
%%%%%%%%%%%%%%%%%%%%%%%%%%

  \label{Spectral triples}

  \begin{definition} \label{DefSpTr}
    A (unital) spectral triple is defined by the data $(\A,\H,\DD)$ with
    \begin{itemize}
      \item[--] an involutive unital algebra $\A$,
      \item[--] a faithful representation $\pi$ of $\A$ on a Hilbert space $\H$,
      \item[--] a selfadjoint operator $\DD$ acting on $\H$ with compact resolvent such that for any $a\in\A$, the extended operator of $\commut{\DD}{\pi(a)}$ is bounded.
    \end{itemize}
    The spectral dimension of the triple is $d\vc\inf\{d'>0\text{ such that }\Tr \abs{\DD}^{-d'} < \infty \}$.\\
    The spectral triple is called regular if the spaces $\A$ and $\commut{\DD}{\A}$ are contained in the domain of $\delta^k$, for all $k\in\N$, where $\delta(a) \vc \commut{\abs{\DD}}{a}$, $a\in\A$.
  \end{definition}

  %%%%%%%%%%%%%%%%%%%%%%%%%%%%%%%%%%%%%%%%%%%%%%%%%%%%%
  \subsection{\texorpdfstring{Spectral triples for Hardy space on $\Omega$}{}}
  %%%%%%%%%%%%%%%%%%%%%%%%%%%%%%%%%%%%%%%%%%%%%%%%%%%%%
    \label{secHardy}

    \begin{prop}
      \label{dimspec}
      For a bounded domain $\Omega$ as in Section \ref{Notations and definitions}, let $\A_H$ be the algebra of all GTO's of order $\leq0$, with the identity representation $\pi$ on $\mathcal{H} \vc H^2(\partial\Omega)$, and $\DD$ be a selfadjoint elliptic generalized Toeplitz operator of order $1$ on $\H$.\\
      Then $(\A_H,\H,\DD)$ is a regular spectral triple of dimension $n=dim_{\C}\,\Omega$.
    \end{prop}

    \begin{proof}
      Clearly $\A_H$ is an algebra with unit $T_1= \bbbone,$ and involution $T_P^* = T_{P^*}$, where $P^*$ is the adjoint of $P$ in $L^2(\partial\Omega)$, and trivially $\pi$ is faithful. Since $\mathcal{D}$ is elliptic of order 1, it has a parametrix of order $-1$, hence compact, so $\mathcal{D}$ has compact resolvent. Moreover, for any $T_P \in \A_H$, the commutator $\commut{\DD}{T_P}$ is bounded since, as commutator of GTO's, $\ord(\commut{\DD}{T_P}) \leq \ord(\DD)+\ord(T_P)-1 \leq 1+0-1=0$. So $(\A_H,\H,\DD)$ is a spectral triple.

      Since $\vert\DD\vert$ is of order 1 (see for instance \cite[Proposition 16]{E2008}), one can check recursively that for all $k\in\N$ and $T_P\in\A_H$, $\delta^k(T_P)=\commut{\abs{\DD}}{T_k}$, where $T_k$ is a GTO of order 0, so the commutator is bounded. The same is true for elements of the form $T_P=\commut{\DD}{T_Q}$, $T_Q\in\A_H$, so the regularity follows.

      For the dimension computation, we follow \cite[Theorem 3]{EZ2010}. We sort the points $\lambda_j$ of the spectrum of $\vert\DD\vert\vc(\DD^*\DD)^{1/2}$ counting multiplicities as $0<\lambda_1\leq\lambda_2\leq \cdots$. Let $M(\lambda)$ be the number of $\lambda_j$'s less than $\lambda$. We can apply \cite[Theorem 13.1]{BMG1981} to $\vert\DD\vert$ which is of order 1:
      \begin{align*}
        M(\lambda) \underset{\lambda\to\infty}{=} \tfrac{\text{vol}(\Sigma_\DD)}{(2\pi)^n}\lambda^n + \mathcal{O}(\lambda^{n-1}),
      \end{align*}
      where, using \eqref{Sigma}, $\Sigma_\DD \vc \{(x,\xi)\in\Sigma: \sigma(\DD)(x,\xi) \leq 1 \, \}$.
      So, if $c\defeql (2\pi)^{-n}\,\text{vol}(\Sigma_D)$, we get for large $\lambda$:
      \begin{align*}
        \lambda^n = \tfrac{M(\lambda)}{c} + \mathcal{O}(\lambda^{n-1}) = \tfrac{M(\lambda)}{c} + \mathcal{O}\big(\lambda^{-1}M(\lambda)\big).
      \end{align*}
      Since $M(\lambda)^{-1/n}\sim\mathcal{O}(\lambda^{-1})$, we have $\lambda^n = \tfrac{M(\lambda)}{c} + \mathcal{O}(M(\lambda)^{1-1/n})= \tfrac{M(\lambda)}c[1+O(M(\lambda)^{-1/n})]$ as $\lambda\to\infty$, so given $d\in\R$,
      \begin{align*}
        \lambda^{-d} = \tfrac{c^{d/n}[1+O(M(\lambda)^{-1/n})]}{M(\lambda)^{d/n}} = \tfrac{c^{d/n}}{M(\lambda)^{d/n}} + \mathcal{O}(\tfrac{1}{M(\lambda)^{(d+1)/n}}).
      \end{align*}
      Thus
      \begin{align*}
        \Tr \abs{\DD}^{-d} = & \sum_{j=1}^\infty \lambda_j ^{-d} = \int_{\lambda_1}^\infty \lambda^{-d} dM(\lambda) = \int_{\lambda_1}^\infty \big( \tfrac{c^{d/n}}{M(\lambda)^{d/n}} + \mathcal{O}(\tfrac{1}{M(\lambda)^{(d+1)/n}}) \big)\,dM(\lambda) \\
        = &\int_{1}^\infty \big(\tfrac{c^{d/n}}{M^{d/n}} + \mathcal{O}(\tfrac{1}{M^{(d+1)/n}}) \big) \,dM
      \end{align*}
      is finite if and only if $d>n$.
    \end{proof}

    \begin{remark}
      If we assume in above proposition that $\DD$ is of order $a<1$, then the commutators with $T_P$ will then be GTOs of order $a-1$, hence not only bounded but even compact.
    \end{remark}

  %%%%%%%%%%%%%%%%%%%%%%%%%%%%%%%%%%%%%%%%%%%%%%%%%%%%
  \subsection{\texorpdfstring{Spectral triples for Bergman space on $\Omega$}{}}
  %%%%%%%%%%%%%%%%%%%%%%%%%%%%%%%%%%%%%%%%%%%%%%%%%%%%
    \label{secBergm}

    In the Bergman case, we have a similar result as \propref{dimspec}:

    \begin{prop}
      \label{TripleBergman}
      For a bounded domain $\Omega$ as in Section \ref{Notations and definitions}, let  $\A_B$ be the algebra generated by the Toeplitz operators $\mathbf{T}_f$, with $f\in\Coo(\overline{\Omega})$, with the identity representation $\pi$ on $\mathcal{H}\vc \bergonw{w}{\Omega}$, and $\mathcal{D} \defeql V_w \, T \, V_w^*$, where $T$ is a selfadjoint elliptic GTO of order $1$ and $V_w$ as in \eqref{V_w}.\\
      Then $(\A_B,\H,\DD)$ is a regular spectral triple of dimension $n=\dim_{\C}\Omega$.
    \end{prop}

    \begin{proof}
      As in the Hardy case, clearly $\A_B$ is a unital involutive algebra with a faithful representation on $\mathcal{H}$. Since $T$ has a parametrix of order $-1$, hence compact, $\DD$ has compact resolvent by unitary equivalence. \\
      To see that $\commut{\DD}{\mathbf{T}_f}$ is bounded for all $\mathbf{T}_f$ in $\A_B$, we use \eqref{upsP} and remark that
      \begin{align}
        \label{commutDTf}
        \commut{\DD}{\mathbf{T}_f} = V_w \, \commut{T}{T_{\Lambda_w}^{-1/2}T_{\Lambda_{wf}} T_{\Lambda_w}^{-1/2}} \, V_w^*.
      \end{align}
      Since the order of the GTOs $T$ and $T_{\Lambda_w}^{-1/2}T_{\Lambda_{wf}} T_{\Lambda_w}^{-1/2}$ are respectively $1$ and less or equal to $0$, the commutator on the right hand side has order less or equal to $0$, hence is bounded in particular on $H^2(\partial\Omega)$.

      Since $\abs{\DD}=V_w\,\abs{T}\,V_w^*$ and $|\DD|^{-s}=V_w \, \abs{T}^{-s} \, V_w^*$, for $s\in\R$, the regularity and dimension computation are shown by using the same arguments as in \propref{dimspec}.
    \end{proof}

 Let $\mathcal T^{-\infty}$ denote the ideal in $\A_H$ of GTO's of order $-\infty$, i.e. of smoothing generalized Toeplitz operators with Schwartz kernel in $C^\infty(\partial\Omega\times\partial\Omega)$.

    \begin{prop}
    The map $\psi:\,a\in\A_B \mapsto V_w^* \,a\,V_w \in\A_H$ is a $^*$-isomorphism of $\A_B$ onto a subalgebra $\psi(\A_B) \subset \A_H$. Moreover, $\A_H=\psi(\A_B) + \mathcal T^{-\infty}$.
    \end{prop}

    \begin{proof}
     Let us first show the inclusion: thanks to \eqref{Tf=VTTTV*}, $\psi (\mathbf{T}_f)=V_w^* \,\mathbf{T}_f \,V_w=T_{\Lambda_w}^{-1/2}\,T_{\Lambda_{wf}} \, T_{\Lambda_w}^{-1/2} \in \A_H$. Since $\A_B$ is generated by the $\mathbf{T}_f $, the map defines a isomorphism from $\A_B$ into $\psi(\A_B)$ which preserves the adjoint.

     For the second part, let $T\in\A_H$, denote by $-s\le0$ the order of $T$ and by $u_0(x')\|\xi'\|^{-s}$, $u_0\in C^\infty(\partial\Omega)$, its principal symbol. If
      $$ f_0(x) := \tfrac{\Gamma(m_w+1)}{\Gamma(m_w+s+1)}\, K(\|\eta\|^s u_0)(x) ,      $$
      then $f_0\in C^\infty(\overline\Omega)$ and by \propref{gammaTK}, \eqref{princ} and \eqref{princP}, the operator $\psi(\mathbf T_{r^s f_0})$ is a GTO also of order $-s$ and with the same principal symbol as $T$; thus $T_1:=T-\psi(\mathbf T_{r^s f_0})$ is a GTO of order $-s-1$. Applying the same reasoning to $T_1$ in the place of $T$ yields $f_1\in C^\infty(\overline\Omega)$ such that $\psi(\mathbf T_{r^{s+1} f_1})$ has the same order and principal symbol as $T_1$, hence $T_2:=T-\psi(\mathbf T_{r^s f_0+r^{s+1}f_1})$ is a GTO of order $-s-2$. Continue in this way to construct $f_2,f_3,\dots$, and then let $f\in C^\infty(\overline\Omega)$ be a function which has the same boundary jet as the formal sum $\sum_{j=0}^\infty r^j f_j$; that is, such that
      $$ f - \sum_{j=0}^k r^j f_j = \mathcal{O}(r^{k+1}) $$
      vanishes to order $k+1$ at the boundary, for any $k=0,1,2,\dots$. (Such function can be obtained in a completely standard manner along the lines of the classical Borel's theorem.) Set $g:=r^s f$. Then by \propref{gammaTK} and \eqref{princ} again, for any $k\in \N$, the difference
      $$
      R\vc T-\psi(\mathbf T_g) = T-\psi(\mathbf T_{\sum_{j=0}^k r^{s+j}f_j}) - \psi(\mathbf T_{r^s(f-\sum_{j=0}^k r^j f_j)})
      = T_{k+1} - \psi(\mathbf T_{\mathcal{O}(r^{k+s+1})})
      $$
      is a GTO of order (at most) $-s-k-1$. Since $k$ is arbitrary, $R$ is a GTO of order $-\infty$, i.e. $R\in\mathcal T^{-\infty}$, and the proof is complete.
    \end{proof}

      For a function $f$ in $\Coo(\overline{\Omega})$ vanishing at order $j\in\N$ on the boundary, the order of $\psi(\mathbf{T}_f)$ is $-j$, since, from \eqref{princ1} and \eqref{princP}, the expression for the principal symbol is
      \begin{align*}
        \sigma(\psi(\mathbf{T}_f))(x',\xi') = \tfrac{\Gamma(m_w+j+1)}{\Gamma(m_w+1)} \,\tfrac{ 2^{-j}}{j!} \, \partial_{\mathbf{n}} ^{j} f(x')\, \norm{\xi'}^{-j} .
      \end{align*}
      Hence all functions $f$ that vanish to infinite order on the boundary (so they are not analytic on the boundary) are such that $\psi(\mathbf{T}_f) \in \mathcal T^{-\infty}$.

    \begin{question}
      Is the inclusion $\psi(\A_B) \subset \A_H$ strict? $($Or, is $\mathcal T^{-\infty}$ not contained in $\psi(\A_B)$?$)$
    \end{question}
    \medskip

    We now give two examples of operator $\DD$ for these Bergman triples:

    Using \propref{propTPpsdo}, we may take $\DD=\mathbf T_{\mathbf P}$, with any differential operator $\mathbf P$ of order 1 on $\overline\Omega$ such that $\mathbf T_{\mathbf P}=\mathbf T_{\mathbf P}^*$ and \eqref{princP} is nonzero on $\Sigma$. So the first example is given in Remark \ref{TPselfadjoint}.

    \bigskip

    As a second example, we deduce from \remref{reminvTf} that $(\mathbf{T}_{r})^{-1}$ exists on $\text{Ran}(\mathbf{T}_{r})$ which is dense in $\bergonw{w}{\Omega}$. So we can take $\DD = (\mathbf{T}_{r})^{-1}$ as an example of Dirac operator on $\bergonw{w}{\Omega}$ and construct the spectral triple $(\mathcal{A}, \mathcal{H}, \mathcal{D})$ with the same $\mathcal{A}$ and $\mathcal{H}$ as in of Proposition \ref{TripleBergman}. However, the positivity of $\DD = (\mathbf{T}_{r})^{-1}$ induces a trivial K-homology class for the spectral triple.
    We now get around this triviality:

    \begin{prop}
      \label{unitaries}
      Let $(\A,\H,\DD)$ be the spectral triple of \propref{TripleBergman} with $\DD=(\mathbf{T}_r)^{-1}$. Define $\widetilde{\A}$ as the algebra of all $ \mathbf{T}_f $'s acting diagonally on $\widetilde{\mathcal{H}} \vc \mathcal{H} \oplus \mathcal{H}$ and let $\widetilde{\mathcal{D}}$ be the operator
      \begin{align*}
        \widetilde{\mathcal{D}} \vc \small{\left(
        \begin{array}{cc}
          0 & \mathbf{U}\, (\mathbf{T}_r)^{-1} \\
          (\mathbf{T}_r)^{-1} \, \mathbf{U}^* & 0
        \end{array} \right)}
      \end{align*}
      where $\mathbf{U}$ is a unitary operator on $\bergonw{w}{\Omega}$. If $\mathbf{U}$ is such that
      \begin{align}
      \label{hypgammaUK}
        V_w^*\, \mathbf{U}\, V_w \text{ is a unitary GTO,}
      \end{align}
      then $(\widetilde{\mathcal{A}},\widetilde{\mathcal{H}},\widetilde{\mathcal{D}})$ is a regular spectral triple.\\
      The triples $(\A,\H,\DD)$ and $(\widetilde{\mathcal{A}},\widetilde{\mathcal{H}},\widetilde{\mathcal{D}})$ have the same dimension.
    \end{prop}

    \begin{proof}
      We first check the boundedness of $\commut{\widetilde{\mathcal{D}}}{\widetilde{\mathbf{T}}_f}$. For any $\widetilde{\mathbf{T}}_f \in \widetilde{\mathcal{A}}$, we have $\commut{\widetilde{\mathcal{D}}}{\widetilde{\mathbf{T}}_f} =\left(\smallmatrix   0 & D_1 \\ D_2 & 0 \endsmallmatrix \right)$
      $D_1 \vc \commut{\mathbf{U}\,\mathbf{T}_r^{-1}}{\mathbf{T}_f}$ and $D_2 \vc \commut{\mathbf{T}_r^{-1}\mathbf{U^*}}{\mathbf{T}_f}$. From \propref{gammaTK}, we have the relations $\mathbf{T}_f  = V_w  \, T_{\Lambda_{w}}^{-1/2}T_{\Lambda_{wf}} T_{\Lambda_{w}}^{-1/2}\, V_w^*$, and $(\mathbf{T}_r)^{-1} = V_w \, T_{\Lambda_{w}}^{1/2}T_{\Lambda_{wr}}^{-1}T_{\Lambda_{w}}^{1/2} \, V_w^*$. We get
      \begin{align*}
        D_1 = & \, \mathbf{U} \mathbf{T}_r^{-1} \, \mathbf{T}_f - \mathbf{T}_f \, \mathbf{U} \mathbf{T}_r^{-1}\\
        = & \, \mathbf{U} \, (V_w  T_{\Lambda_w}^{1/2}T_{\Lambda_{wr}}^{-1}T_{\Lambda_w}^{1/2} V_w^*) \, (V_w  T_{\Lambda_{w}}^{-1/2}T_{\Lambda_wf}T_{\Lambda_w}^{-1/2} V_w^*)\\
        & \hspace{5.1cm} - (V_w  T_{\Lambda_{w}}^{-1/2}T_{\Lambda_{wf}}T_{\Lambda_w}^{-1/2} V_w^*) \, \mathbf{U} \, (V_w  T_{\Lambda_w}^{1/2}T_{\Lambda_{wr}}^{-1}T_{\Lambda_w}^{1/2} V_w^*) \\
        = & \, \mathbf{U} \, (V_w  T_{\Lambda_w}^{1/2}T_{\Lambda_{wr}}^{-1}T_{\Lambda_wf}T_{\Lambda_w}^{-1/2} V_w^*) - V_w  T_{\Lambda_{w}}^{-1/2}T_{\Lambda_{wf}}T_{\Lambda_w}^{-1/2} \, (V_w^*\mathbf{U}V_w)\, T_{\Lambda_w}^{1/2}T_{\Lambda_{wr}}^{-1}T_{\Lambda_w}^{1/2} V_w^*\\
        = & \, (V_wV_w^*) \, \mathbf{U} \, V_w  T_{\Lambda_w}^{1/2}T_{\Lambda_{wr}}^{-1}\,(T_{\Lambda_w}^{1/2}T_{\Lambda_w}^{-1/2})\,T_{\Lambda_wf}T_{\Lambda_w}^{-1/2} V_w^*\\
        & \hspace{5.1cm} - V_w  T_{\Lambda_{w}}^{-1/2}T_{\Lambda_{wf}}T_{\Lambda_w}^{-1/2} \, (V_w^*\mathbf{U}V_w)\, T_{\Lambda_w}^{1/2}T_{\Lambda_{wr}}^{-1}T_{\Lambda_w}^{1/2} V_w^*\\
        = & V_w \, \commut{(V_w^*\mathbf{U}V_w)\,T_{\Lambda_w}^{1/2}T_{\Lambda_{wr}}^{-1}T_{\Lambda_w}^{1/2}}{T_{\Lambda_w}^{-1/2}T_{\Lambda_{wf}}T_{\Lambda_w}^{-1/2}} \, V_w^*\,.
      \end{align*}
      From the hypothesis, $V_w \mathbf{U} V_w^* $ is a bounded GTO, $T_{\Lambda_{wr}}^{-1}T_{\Lambda_w}$ is a GTO of order 1 and $T_{\Lambda_w}^{-1}T_{\Lambda_{wf}}$ is a GTO of order less or equal to 0, so the commutator is a GTO of order less or equal to 0, thus is a bounded operator on $\bergonw{w}{\Omega}$. Similar arguments show that
      \begin{align*}
        D_2 = V_w \, \commut{ T_{\Lambda_w}^{1/2}T_{\Lambda_{wr}}^{-1} T_{\Lambda_w}^{1/2}\,(V_w^* \mathbf{U}^* V_w )}{T_{\Lambda_w}^{-1/2}T_{\Lambda_{wf}}T_{\Lambda_w}^{-1/2}} \, V_w^*
      \end{align*}
      is also bounded on $\bergonw{w}{\Omega}$, which makes $\commut{\widetilde{\mathcal{D}}}{\widetilde{\mathbf{T}}_f}$ bounded on the direct sum $\widetilde{\mathcal{H}}$.

      We remark that the expression of $D_1$ and $D_2$ differ from \eqref{commutDTf} by the term $V_w^* \mathbf{U} V_w$ which is a GTO of order 0. So the regularity of the spectral triple is proven as in \propref{TripleBergman}.

      Finally $\widetilde{\mathcal{D}}$ has compact resolvent since $\widetilde{\mathcal{D}}^{-1} = \left( \smallmatrix 0 & \mathbf{U}\mathbf{T}_r \\ \mathbf{T}_r \mathbf{U}^* & 0\endsmallmatrix \right)$ is compact because the operators $\mathbf{U}\mathbf{T}_r$ and $\mathbf{T}_r\mathbf{U}^*$ are compact.

      Since $\widetilde{\mathcal{D}}^2 = \left( \smallmatrix \mathbf{U}\mathbf{T}_r^{-2}\mathbf{U}^* & 0\\ 0 & \mathbf{T}_r^{-2}\endsmallmatrix \right)$, we deduce that the unitary $\mathbf{U}$ does not influence the calculation of eigenvalues.
    \end{proof}

    \begin{remark}
      \label{UTrem}
      {\rm
      The two classes of unitaries $\mathbf U$ defined in \remref{Unitary-Rem} satisfy \eqref{hypgammaUK}, so provide examples of spectral triples $(\widetilde{\mathcal{A}},\widetilde{\mathcal{H}},\widetilde{\mathcal{D}})$ on (the sum of two copies of) the Bergman space with non-positive $\widetilde{\mathcal{D}}$ when $\DD=(\mathbf{T}_r)^{-1}$.
      }
    \end{remark}

    For the case of the unit ball with a radial weight, the \propref{TripleBergman} can be made much more explicit. Indeed,
    if $f$ is a radial function in $\Coo(\overline{\mathbb{B}^n})$ and the weight $w$ is as in \eqref{w}, the family $\{v_\alpha\}_{\alpha\in\N^n}$ defined in \eqref{valpharad} diagonalizes $\mathbf{T}_f : A^2_w(\mathbb{B}^n) \to A^2_w(\mathbb{B}^n)$ and the eigenvalues only depend on $\abs{\alpha}$.
    Namely,
    \begin{align*} \label{lambdaf}
      \scalp{\mathbf{T}_f \,v_\alpha}{v_\beta}_{A^2_w} =  \tfrac{\delta_{\alpha\beta}}{\int_0^1 t^{2n+2\abs{\alpha} -1} \,w(t) \, dt }\,\int_0^1 t^{2n+2\abs{\alpha}-1} \,f(t) \,w(t)\, dt,
    \end{align*}
    as is easily seen by passing to the polar coordinates. \\
    Thus, assume that $w=r^{m_w}$, $m_w\in\N$ where the function $r$ of the form \eqref{r} depends only on the variable $\abs{x}$, for $x$ in $\mathbb{B}^n$. For convenience, we temporarily denote here $\partial \vc \partial_{\mathbf{n}}$. Then, since the first $m_w-1$ derivatives of $w$ vanish on $\partial\mathbb{B}^n=S^{2n-1}$, and $\partial^{k}(w\,r)$ is non-zero only for $k > m_w$, we have for $\alpha \in \N^n$:
    \begin{align*}
      & \int_0^1 w(t)\,t^{2n+2\abs{\alpha}-1}\,dt  =   \tfrac{(-1)^{m_w} \, \partial^{m_w} w(1)}{ \prod_{k=0}^{m_w}(2n+2\abs{\alpha} + k)} - \tfrac{(-1)^{m_w}}{\prod_{k=0}^{m_w}(2n+2\abs{\alpha} + k)} \int_0^1 \partial^{m_w+1} w(t) \,t^{2n+2\abs{\alpha}+m_w} \,dt,\\
      & \int_0^1 w(t)\,r(t)\,t^{2n+2\abs{\alpha}-1}\,dt  = \tfrac{(-1)^{m_w+1} \, \partial^{m_w+1} (w\,r)(1)}{ \prod_{k=0}^{m_w+1}(2n+2\abs{\alpha} + k)} \\
      & \hspace{6cm}- \tfrac{(-1)^{m_w+1}}{\prod_{k=0}^{m_w+1}(2n+2\abs{\alpha} + k)} \int_0^1 \partial^{m_w+2} (w\,r)(t) \, t^{2n+2\abs{\alpha}+m_w+1} \,dt.
    \end{align*}
    Since $\partial^{m_w} w(1)\neq0$ by hypothesis, we deduce, applying the Leibniz formula for $\partial^{m_w+1}(w\,f)(1)$,
    \begin{align*}
      & \int_0^1 w(t)\,t^{2n+2\abs{\alpha}-1}\,dt \underset{\abs{\alpha}\to \infty}{\sim}  \tfrac{(-1)^{m_w} \, \partial^{m_w} w(1)}{ \prod_{k=0}^{m_w}(2n+2\abs{\alpha} + k)}, \\
      & \int_0^1 w(t)\,r(t)\,t^{2n+2\abs{\alpha}-1}\,dt  \underset{\abs{\alpha}\to \infty}{\sim} \tfrac{(-1)^{m_w+1} \, \partial^{m_w}w(1) \, \partial r(1)}{ \prod_{k=0}^{m_w+1}(2n+2\abs{\alpha} + k)}.
    \end{align*}
    We choose now $\DD=\mathbf{T}_r$ and from the previous result, $\DD$ diagonal in the basis $\{v_\alpha\}_{\alpha\in\N^n}$ and we obtain the asymptotic behavior of the eigenvalues:
    \begin{align*}
      \scalp{\DD\,v_\alpha}{v_\alpha}_{A^2_w} \underset{\abs{\alpha}\to \infty}{\sim} - \tfrac{1}{2n+2\abs{\alpha} + m_w+1} \,\partial r(1)\underset{\abs{\alpha}\to \infty}{\sim} - \tfrac{1}{2\abs{\alpha}}\,\partial r(1).
    \end{align*}
    Since $\Tr \abs{\mathcal{D}}^{-d'}=\sum_{k=0}^\infty  \binom{n-1+k }{n-1}\,\tfrac{\partial r(1)}{(2k)^{d'}}$ and $ \binom{n-1+k }{n-1} \underset{k\to \infty}{\sim} \tfrac{k^{n-1}}{(n-1)!}$, we have $\Tr \abs{\mathcal{D}}^{-d'} <\infty$ if and only if $\tfrac{1}{(n-1)!} \sum_{k=0}^{\infty} \tfrac{k^{n-1}}{k^{d'}} <\infty$, so for each $d'>d=n$.

    To give an example, we choose the function $r: z\in \Omega \mapsto 1-\abs{z}^2$, and the weight $w=r^0=1$. A direct calculation shows that $\{v_\alpha\}_{\alpha\in\N^n}$ defined in \eqref{valpha} diagonalizes the operator $\mathbf{T}_{1-\abs{z}^2}$, acting on $\berg$ with eigenvalues $\lambda_\alpha \defeql \tfrac{1}{n+\abs{\alpha} + 1}$ and multiplicity $\binom{n-1+\abs{\alpha} }{n-1}$. Indeed:
    \begin{align*}
      \mathbf{T}_{1-\abs{z}^2} \,v_\alpha  =\, & v_\alpha -  \sum_{\beta \in \N^n} \sum_{j=1}^n \scalp{z_j\bar{z}_j \,v_\alpha}{v_\beta}\, v_\beta = v_\alpha - \sum_{\beta \in \N^n}  \sum_{j=1}^n \scalp{b_\alpha z^{\alpha + 1_j}}{b_\beta z^{\beta + 1_j}}\, v_\beta\\
      =\, & (1 - \sum_{j=1}^n b_\alpha^2 \|z^{\alpha+1_i}\|^2 ) \,v_\alpha= (1 - \sum_{j=1}^n\tfrac{b_\alpha^2}{b_{\alpha+1_j}^2}) \,v_\alpha.
    \end{align*}

    Since $b_{\alpha+1_j}^2 = \tfrac{(n+\abs{\alpha} +1 )!}{n!\alpha!(\alpha_j +1)\,\mu(\mathbb{B}^n)}$, we get $\lambda_\alpha = 1-\sum_{j=1}^n \tfrac{(n+\abs{\alpha})!}{n!\alpha!}\tfrac{n!\alpha!(\alpha_j+1)}{(n+\abs{\alpha} +1)!} = 1-\tfrac{n+\abs{\alpha}}{n+\abs{\alpha} +1 } = \tfrac{1}{n+\abs{\alpha}+1} \,,$
    and the multiplicity follows because $\lambda_\alpha$ depends only of $\abs{\alpha}$ and we see directly the appearance of the  dimension $n$.

  %%%%%%%%%%%%%%%%%%%%%%%%%%%%%%%%%%%%%%%%%%%%%%%%%%%%
  \subsection{\texorpdfstring{Example of spectral triples on the unit ball of $\C^n$ without weight}{}}
  %%%%%%%%%%%%%%%%%%%%%%%%%%%%%%%%%%%%%%%%%%%%%%%%%%%%
    \label{secHeisen}

    We consider in this section the model case $\Omega=\mathbb{B}^n$ with $w=1$.

    \begin{prop}
      \label{triplesHeisenberg}
      Let $\A = \{T_u , u\in \Coo(S^{2n-1}) \}$ be the algebra of Toeplitz operators on the Hardy space $\H=H^2(S^{2n-1})$ and $\DD_{S^{2n-1}}$ the operator from \eqref{newD2}. Then $(\A,\H,\DD_{S^{2n-1}})$ is a regular spectral triple of dimension $2n$.

      Let $\A = \{\mathbf{T}_f , f\in \Coo(\overline{\mathbb{B}^n}) \}$ be the algebra of Toeplitz operators on the Bergman space $\H=\bergonw{}{\mathbb{B}^n}$ and $\DD_{\mathbb{B}^n}$ the operator from \eqref{newD1}. Then $(\A,\H,\DD_{\mathbb{B}^n})$ is a spectral triple of dimension $2n$.
    \end{prop}

    \begin{proof}
      The requirement of compact resolvent is fulfilled automatically, since it is fulfilled for the standard Dirac operator on $\R^n$, from which $\DD_{S^{2n-1}}$ and $\DD_{\mathbb{B}^n}$ were obtained by transferring via various $^*$-isomorphisms, which also shows they are selfadjoint. \\
      From \lemref{gammatauPjK}, $\DD_{S^{2n-1}} = (t/2)^{1/2} \sum_j {\text{\footnotesize $\Gamma_j$}}\,T^{1/2}_{\Lambda} \traceOp \tau_t(P_j) K T^{-1/2}_{\Lambda}$, is a GTO of order $1/2$ and we use the same argument as in the proof of \propref{dimspec}. The result for the Bergman case follows from the identity $\DD_{\mathbb{B}^n} = V \DD_{S^{2n-1}} V^*$ and using a similar reasoning as in the proof of \propref{TripleBergman}.
   \end{proof}

  %%%%%%%%%%%%%%%%%%%%%%%%%%%%%%%%%%%%%%%%%%%%%%
  \subsection{\texorpdfstring{Dixmier traces}{}}
  %%%%%%%%%%%%%%%%%%%%%%%%%%%%%%%%%%%%%%%%%%%%%%

    In all the examples of spectral triples above, one can also give a formula for the Dixmier traces $\Trw(a|\DD|^{-d})$ where $a\in\A$ and $d$ is the spectral dimension.

    First, recall that by \cite[Theorem 3]{EZ2010}, if $T_P$ is a GTO on $\partial\Omega$ of order $-n$, then $T_P$ is in the Dixmier-class, is measurable, and
    \begin{equation}\label{EZf}
      \Trw(T_P) = \tfrac1{n!(2\pi)^n} \int_{\partial\Omega} \sigma(T_P)(x',\eta_{x'}) \, \nu_{x'}\,.
    \end{equation}
    This formula is independent of the choice of the defining function $r$ (see \cite[Remark 4]{EZ2010}).

    In the context of the Hardy space spectral triple from Section \ref{secHardy}, we thus have for any $u\in C^\infty(\partial\Omega)$ and $\DD$ as in \propref{dimspec}
    \begin{equation}
      \label{dixmierHardy}
      \Trw(T_u|\DD|^{-n}) = \tfrac1{n!(2\pi)^n} \int_{\partial\Omega} u(x') \,  |\sigma(\DD)(x',\eta_{x'})|^{-n} \, \nu_{x'}\,,
    \end{equation}
    and similarly for $T_u$ replaced by any GTO $T_Q$ of order 0: the $u(x')$ in the integrand is then replaced by $\sigma(T_Q)(x',\eta_{x'})$.

    For the Bergman case, the Dirac operator in \propref{TripleBergman} is of the form $\DD = V_w T V_w^*$, where $T$ is a selfadjoint elliptic GTO of order 1. So we have for any $f\in C^\infty(\overline\Omega)$
    \begin{align*}
      \Trw(\mathbf T_f|\DD|^{-n}) & = \Trw(V_w\,T_{\Lambda_w}^{-1/2}\,T_{\Lambda_wf}\,T_{\Lambda_w}^{-1/2}\,V_w^* \, V_w \, \abs{T}^{-n}\,V_w^*) \\
      &= \Trw(T_{\Lambda_w}^{-1/2}\,T_{\Lambda_wf}\,T_{\Lambda_w}^{-1/2} \, \abs{T}^{-n}),
    \end{align*}
    which is treated as above.\\
    For $\DD=\mathbf T_{\mathbf {P_n}}, \mathbf{P_n}=\sum_j \overline{\partial_j r} \partial_j$ as in \remref{TPselfadjoint}, we use a similar trick to compute
    \begin{align*}
      \Trw(\mathbf T_f|\DD|^{-n}) &=  \Trw(V\,T_{\Lambda}^{-1/2}\,T_{\Lambda_f}\,T_{\Lambda}^{-1/2}\,V^* \, V\, T_{\Lambda}^{1/2} \, (\gamma \mathbf{T_{P_n}} K) \, T_{\Lambda}^{-1/2} \, V^*)\\
      &= \Trw(T_{\Lambda}^{-1/2} \, T_{\Lambda_wf} \, (\gamma \mathbf{T_{P_n}} K) \, T_{\Lambda}^{-1/2}),
    \end{align*}
    and we get the result from \eqref{gammasigmaTPnK}.

    For the triple $(\widetilde \A, \widetilde \H, \widetilde \DD )$ from \propref{unitaries}, the Dixmier traces get multiplied by 2 due to the appearance of $2\times2$ block matrices (see the last paragraph of the proof of the proposition).

    Similarly, a factor $n$ appears in the computation of Dixmier traces in the context of Section \ref{secHeisen} since the Dirac operators $\DD_{S^{2n-1}}$ and $\DD_{\mathbb{B}^n}$ involved in \propref{triplesHeisenberg} contain gamma matrices.
    From the identity $V^*\, \tau(P_j)\, V = T^{1/2}_{\Lambda}\, (\gamma\, \tau(P_j)\,K) \, T^{-1/2}_{\Lambda}$ and \lemref{gammatauPjK}, we remark that $\DD_{S^{2n-1}} = (t/2)^{-1/2}\sum_j \text{\footnotesize $\Gamma_j$} T_{Q_j}$, where the $T_{Q_j}$ are GTOs of order 1/2 whose symbols are known. Hence $\Trw(|\DD_{S^{2n-1}}|^{-2n}) = \Trw(|\DD_{\mathbb{B}^n} |^{-2n})$ is finite. We use \eqref{dixmierHardy} again to compute $\Trw(T_u|\DD_{S^{2n-1}}|^{-2n})$ and the Bergman case follows as above.

    \begin{remark}
      {\rm
      Note that pseudodifferential operators of order $k$ on $\R^n$ are transformed in GTOs of order $k/2$ on the boundary of $\Omega$ (in the beginning of Section 4), which might seem to be at odds with the fact that Dixmier-trace operators correspond to both pseudodifferential operators and GTOs of order $-n$, respectively on a compact real manifold of dimension $n$ and on the boundary of a complex domain of dimension $n$. The point is that on a compact manifold of dimension $n$ this is true, but on $\R^n$ this fails: for instance the operator $(\bbbone-\Delta)^{-n/2}$ on $\R^n$ is not even compact, much less in Dixmier-class. What one needs is first of all not to use H\"ormander but Shubin (also known as Grossman--Loupias--Stein) classes of pseudodifferential operators (i.e. with prescribed decay of symbols not only as $\xi$ goes to infinity, but as $(x,\xi)$ goes to infinity), and secondly, the order needed for the Dixmier-class is then not $-n$ but $-2n$ (see for instance \cite[Theorem 4.1]{BEY} where 
the result is stated for pseudodifferential operators on $\R^n$ of Weyl type, but it is the same for the Kohn--Nirenberg type). Since $(-2n)/2=-n$, this yields precisely the correct order for GTOs, and the contradiction disappears. \\
      Actually, this can be recast the following way: for a non unital spectral triple $(\A,\H,\DD)$, the axiom ``$\DD$ has a compact resolvent'' is replaced by ``$\pi(a) (\bbbone+\DD^2)^{-1}$ is a compact operator for any $a\in \A$''. For instance, in a triple on $\R^n$ like $\big(\text{Functions}(\R^n),\,\H=L^2(\R^n)\otimes \mathcal{C}_n,\,\DD=-i \sum_j{\footnotesize \text{ $\Gamma_j$}} \,\partial_{x_j} \big)$, where $\text{Functions}(\R^n)$ is a subalgebra of $C^\infty(\R^n)$, one can choose for $\A$ the Schwartz space on $\R^n$ to secure this property. While $(\bbbone+\DD^2)^{-n/2}=\big((\bbbone-\Delta)\otimes \bbbone_{\mathcal{C}_n}\big)^{-n/2}$ is not Dixmier-traceable, $\pi(f)(\bbbone+\DD^2)^{-n/2}$ is, so the dimension $n$ appears twice: one in the power of $\vert \DD \vert$  and the other through the algebra $\A$ (via the $n$ variables of $f$).
      }
    \end{remark}

%%%%%%%%%%%%%%%%%%%%%%%%%%%%%%%%%%%%%%%%%
\section{Berezin--Toeplitz star products}
%%%%%%%%%%%%%%%%%%%%%%%%%%%%%%%%%%%%%%%%%

  One may, in a sense, glue the spectral triples from \S{\ref{secBergm}} with different weights $w$ into a single ``composed'' spectral triple, much as Toeplitz operators on weighted Bergman spaces are ``glued'' together in the Berezin--Toeplitz quantization \cite{BMS} \cite{SchliHAB}; this actually yields a spectral triple directly related to the standard Berezin--Toeplitz star product on $\Omega$. Let us give the details.

  Assume that $\log 1/r$ is strictly plurisubharmonic on $\Omega$ (defining functions $r$ with this property exist in abundance due to the strict pseudoconvexity of $\Omega$), so that $ g_{j\overline k}(z) \vc \partial_j\overline \partial_k \log\tfrac1{r(z)}$ defines a K\"ahler metric on $\Omega$; and let $ g\vc r^{n+1}\det[\,g_{j\overline k}\,]$.

  By elementary matrix manipulations, one can check that $g\in C^\infty(\OOm)$ and (thanks to strict pseudoconvexity)
  $g>0$ on $\pOm$. (In fact $g$ coincides with the Monge--Ampere determinant $g=-\det\, [\smallmatrix r&\partial r\,\,\\ \dbar r&\partial\dbar r\endsmallmatrix]$.) \\
  Consider the weighted Bergman spaces $A^2(\Omega,w_m)$ with $w_m\vc r^m g$, $m\in \N$, which we will now denote by $A^2_m$ for brevity. Let
  $$
    \HH \vc \bigoplus_{m=0}^\infty A^2_m
  $$
  be their orthogonal direct sum, and let $\pi_m$ stand for the orthogonal projection in $\HH$ onto the $m$-th summand.
  Denote by $\mathbf N$ the ``number operator'' $ \mathbf N := \bigoplus_m (m+n+1)\,\pi_m$. \\
  For $f\in C^\infty(\OOm)$, we then have the orthogonal sums
  $$
    \TTo_f := \bigoplus_m (\mathbf T_f \text{ on } A^2_m)
  $$
  of the Toeplitz operators $\mathbf T_f$ from Section 2, acting on $\HH$.
  Clearly each $\TTo_f$ is again bounded with $\|\TTo_f\|\le\|f\|_\infty$, $(\TTo_f)^*=\TTo_{\overline f}$, and $[\TTo_f,\pi_m]=0$ for all $m$.

  Let $\cB$ denote the subset of all bounded linear operators $M$ on $\HH$
  for which $[M,\pi_m]=0$ for all $m$ and which possess an asymptotic expansion
  \begin{equation} \label{tTM}
    M \approx \sum_{j=0}^\infty \mathbf N^{-j} \,\TTo_{f_j}
  \end{equation}
  with some $f_j\in C^\infty(\OOm)$ (depending on $M$). Here ``$\approx$'' means
  that
  $$
    \Big\| \pi_m\big(M-\sum_{j=0}^{k-1}\mathbf N^{-j}\,\TTo_{f_j}\big)\,\pi_m\Big\|=O(m^{-k}) \quad\text{as } m\to+\infty   \text{ for any }k=0,1,2,\dots.
  $$

  It is the main result of the Berezin--Toeplitz quantization on $\Omega$ that
  finite products of $\TTo_f$ belong to $\cB$. More specifically, one has \cite{BMS}
  $$
    \TTo_f \TTo_g \approx \sum_{j=0}^\infty \mathbf N^{-j}\, \TTo_{C_j(f,g)}
  $$
  where
  $$
    \sum_{j=0}^\infty h^j \, C_j(f,g) \cv f\star g
  $$
  defines a star product on $(\Omega,g_{j\overline k})$. Symbolically, we can write
  $$
    \TTo_{f\star g}=\TTo_f\, \TTo_g.
  $$
  Another result is, incidentally, that
  \begin{equation}  \label{tTX}
    \| \pi_m\,\TTo_f\,\pi_m \| \to \|f\|_\infty \quad\text{as }m\to+\infty,
  \end{equation}
  implying, in particular, that for a given $M\in\cB$ the sequence $\{f_m\}$ in \eqref{tTM} is determined uniquely.

  There is a neat representation for this whole situation, as follows. (See e.g. \cite{Ecmp}, p. 235, for details in this setting;
  the idea however goes back to Forelli and Rudin.) Consider the ``unit disc bundle'' over $\Omega$:
  $$
    \wO \vc\{(z,t)\in\Omega\times\CC: |t|^2<r(z) \}.
  $$
  The fact that $r$ is a defining function for $\Omega$ implies that $\wO$
  is smoothly bounded, and the facts that $\Omega$ is strictly pseudoconvex
  and $\log 1/r$ is strictly plurisubharmonic imply that $\wO$ is strictly
  pseudoconvex. Thus we have the Hardy space $\wH\vc H^2(\wO)$ of $\wO$ and the
  GTOs $\wT_P$ there, whose symbols $P$ are now pseudodifferential operators on $\partial\wO$.
  A function in $\wH$ has the Taylor expansion in the fiber variable
  $$
    f(z,t) = \sum_{m=0}^\infty f_m(z)\, t^m.
  $$
  Denote by $\wH_m$, $m\in \N$, the subspace in $\wH$ of functions of
  the form $f_m(z)\,t^m$ (i.e. for which all the Taylor coefficients vanish except
  for the $m$-th); alternatively, $\wH_m$ is the subspace of functions in $\wH$ satisfying
  $$
    f(z,e^{i\theta}t) = e^{mi\theta}f(z,t), \quad\forall\theta\in\RR.
  $$
  Then the correspondence
  $$
    f_m(z) \,t^m \longleftrightarrow f_m(z)
  $$
  is an isometric (up to a constant factor) isomorphism of $\wH_m$ onto $A^2_{m-n-1}$. Thus (note that $A^2_{m-n-1}=\{0\}$ for $m\le n$)
  $$
    \wH = \bigoplus_{m=0}^\infty \wH_{m+n+1} \cong \bigoplus_{m=0}^\infty A^2_m = \HH.
  $$
  Furthermore, viewing a function $f\in C^\infty(\Omega)$ also as the function
  $f(z,t):=f(z)$ on $\partial\wO$ (i.e. identifying $f$ with its pullback via
  the projection map), one has, under the above isomorphism,
  $$
    \wT_f \cong \bigoplus_m \,(\mathbf T_f \text{ on }A^2_m) = \TTo_f.
  $$
  Finally, let $\wt K$ be the Poisson operator for $\wO$, and set as before
  $\wt\Lambda:=\wt K{}^*\wt K$. Thus $\wt\Lambda$ is a pseudodifferential operator on $\partial\wO$
  of order $-1$, and a positive selfadjoint compact operator on $\wH$.
  Since the fiber rotations $(z,t)\mapsto(z,e^{i\theta}t)$, $\theta\in\RR$,
  preserve holomorphy and harmonicity of functions, both $\wt K$, $\wt\Lambda$
  and the Szeg\"o projection $\wt S:L^2(\partial\wO)\to\wH$ must commute
  with them. The GTO $\wT_{\wt\Lambda}$ on $\wH$ therefore likewise commutes
  with these rotations, and hence commutes also with the projections in $\wH$
  onto $\wH_m$, i.e. is diagonalized by the decomposition $\wH=\bigoplus_m\wH_m$.
  \\
  Let $L\vc\bigoplus_m L_m$ be the operator corresponding to $\wT_{\wt\Lambda}$
  under the isomorphism $\wH\cong\HH=\bigoplus_m A^2_m$.

  \begin{prop}
    Let $\A$ be the algebra (no closures taken) generated by $\TTo_f$,
    $f\in C^\infty(\OOm)$ acting (via identity representation) on $\H\vc \HH$ and  $\DD \vc L^{-1}$.
    \\
    Then $(\A,\H,\DD)$ is a spectral triple of dimension $n+1$.
  \end{prop}

  \begin{proof}
    Using the above isomorphisms, we can actually switch from $\HH$
    to the space $\wH$, from $\cA$ to the algebra generated by $\wT_f$, $f\in
    C^\infty(\OOm)$ (identified via pullback with functions on $\partial\wO$),
    and to $\cD$ equal to $\wt\Lambda{}^{-1}$. Everything then follows in exactly
    the same way as in Section \ref{secHardy} noting that $\dim_{\C}\wO=n+1$
    (in fact, it is even the special case of the result from that section for
    functions $f$ on $\partial\wO$ that are pullbacks of functions on $\OOm$).
  \end{proof}

  Using \eqref{tTM}, we now alternatively define $\cA$ in the last result
  as a certain subalgebra of formal power series generated by $f\in C^\infty
  (\OOm)$ with the product $\star$, taking for $\pi$ the representation $f\mapsto
  \TTo_f$. More specifically, let $\kappa$ be the linear map from $\cB$ into
  the ring of formal power series
  $$
    \cN \vc C^\infty(\OOm)[[h]]
  $$
  (equipped with the usual involution $\big(\sum_m h^m \,f_m(z)\big)^*:=\sum_m h^m \,\overline{f_m(z)}$ ) given by
  $$
    \kappa: M \longmapsto \sum_{m=0}^\infty h^m f_m(z)
  $$
  for $M$ as in \eqref{tTM}. As noted previously, $\kappa$ is well defined owing
  to \eqref{tTX} (although it is not injective), and, extending as usual $\star$ from
  functions to all of $\cN$ by $\C[[h]]$-linearity,
  $$
    \kappa(MN) = \kappa(M) \star \kappa(N), \quad \kappa(M^*)=\kappa(M)^*,
  $$
  i.e. $\kappa:(\cB,\circ)\to(\cN,\star)$ is a *-algebra homomorphism.
  Then we have the following:

  \begin{theorem}
    Let $\A$ be the subalgebra over $\C[h]$ (no closures taken) of $(\cN,\star)$ generated by $\kappa(\TTo_f)$, $f\in C^\infty(\OOm)$ endowed with the representation $\pi$ on $\H\vc \HH$ be determined by
    \begin{equation} \label{tTP}
      \pi(h^m f) \vc \mathbf N^{-m} \,\TTo_f , \quad f\in C^\infty(\OOm), \,m=0,1,2,\dots,
    \end{equation}
    which is well-defined from $\cA$ into $\cB$, and $\DD\vc\bigoplus_m L_m^{-1}$ on $\HH$.\\
    Then $(\cA,\cH,\cD)$ is a spectral triple of dimension $n+1$.
  \end{theorem}

  \begin{proof}
    In view of the preceding result, the only thing we need to check is that $\pi$ is well-defined and faithful. The former is immediate from
    \eqref{tTP} and the fact that $\kappa:(\cB,\circ)\to(\cN,\star)$ is a \hbox{*-algebra} homomorphism. For the faithfulness, note that
    $\kappa\circ\pi=\operatorname{id}$ on $\cA$; thus $\pi(a)=0$ implies $a=\kappa\big(\pi(a)\big)=0$.
  \end{proof}

  Again, proceeding as in \propref{unitaries}, one can adjoin to the last construction an appropriate unitary GTOs on $\partial\wO$ to obtain also non-positive operators $\widetilde\cD$ (cf. Remark \ref{UTrem}).

%%%%%%%%%%%%%%%%%%%%%%%%%%%%%%%%%%%%%%%%%%
\section*{Appendix}
%%%%%%%%%%%%%%%%%%%%%%%%%%%%%%%%%%%%%%%%%%

  %%%%%%%%%%%%%%%%%%%%%%%%%%%%%%%%%%%%%%%%%%%%%%%
  \section*{Proof of Proposition \ref{Wtunitary}}
  %%%%%%%%%%%%%%%%%%%%%%%%%%%%%%%%%%%%%%%%%%%%%%%

    \renewcommand{\theequation}{A.\arabic{equation}}
    \setcounter{equation}{0}

    \begin{proof} (see also \cite{Zhu2012})
      The integral in $W_\cstFock$ converges since the function $x\mapsto e^{(-z^2 + 2\sqrt{2}xz - x^2)/2\cstFock}$ is square integrable or each $z\in\C^n$. Now, take a closed curve $\Gamma = \Gamma_1\times \dots \times\Gamma_n$ in $\C^n$ and compute:
      \begin{align*}
        \oint_{z\in\Gamma} (W_\cstFock f)(z) \,dz = & \, (\pi\cstFock)^{-n/4} \oint_{z\in\Gamma} \int_{x\in \R^n} e^{(-z^2 + 2\sqrt{2}\,xz - x^2)/2\cstFock} \, f(x)\, dx \, dz.
      \end{align*}
      By Fubini's theorem, we can invert the integrals and get
      \begin{align*}
        \oint_{z\in\Gamma} (W_\cstFock f)(z)\, dz = & \, (\pi\cstFock)^{-n/4} \int_{x\in \R^n} e^{-x^2/2\cstFock} f(x) \big[\oint_{z\in\Gamma} e^{(-z^2 + 2\sqrt{2}\,xz)/2\cstFock} \, dz \big]\, dx.
      \end{align*}
      For a fixed $x\in\R^n$, the analyticity of each $z_k \mapsto e^{(-z_k^2 + 2\sqrt{2}\,x_kz_k)/2\cstFock}$, $k=1\dots n$, induces the vanishing of the integral over $\Gamma$, and so $\oint_{z\in\Gamma} (W_\cstFock f)(z) \,dz = 0$. According to the Morera's theorem, the map $z\mapsto (W_\cstFock f)(z)$ is holomorphic.

      The fact that the Segal--Bargmann transform is a unitary will follow by showing that $W_\cstFock$ maps the orthonormal basis of $L^2(\R^n)$ given by the Hermite functions, to the orthonormal basis of the Fock space we have introduced. To complete the proof, we define the operator $A_j$ on $L^2(\R^n)$, with domain the usual Sobolev space $H^1(\R^n)$, and its Hilbert adjoint $A_j^*$:
      \begin{align*}
        A_j \vc \tfrac{1}{\sqrt{2}}(x_j + \cstFock\,\partial_{x_j}) \, \text{ and } \, A_j^* \vc \tfrac{1}{\sqrt{2}}(x_j - \cstFock\,\partial_{x_j}),
      \end{align*}
      and first show the following result:
      \begin{align} \label{WA_zW}
        W_\cstFock\, A_j^* = z_j\, W_{\cstFock},  \text{ for }j=0,\dots ,n,
      \end{align}
      which follows using integration by parts: for $f\in H^1(\R^n)$,
      \begin{align*}
        \big(W_\cstFock( A_j^* f)\big)(z)  = & \, (\pi\cstFock)^{-n/4}\int_{x\in \R^n} e^{(-z^2 + 2\sqrt{2}\,xz - x^2)/2\cstFock} \, \tfrac{1}{\sqrt{2}}(x_j - \cstFock\partial_{x_j})\,f(x)\, dx \\
        = & \, (\pi\cstFock)^{-n/4}\int_{x\in \R^n} e^{(-z^2 + 2\sqrt{2}\,xz - x^2)/2\cstFock} \, \tfrac{1}{\sqrt{2}} \big(x_j + \cstFock\tfrac{1}{2\cstFock} ( 2\sqrt{2}z_j - 2x_j)\big) \, f(x) \, dx \\
        = & \, (\pi\cstFock)^{-n/4}\int_{x\in \R^n} e^{(-z^2 + 2\sqrt{2}\,xz - x^2)/2\cstFock} \, z_j \, f(x) \, dx = z_j \, (W_{\cstFock}f)(z).
      \end{align*}
      Defining now the operator $A \vc \prod_{j=1}^n A_j$ (the order is not important since the $A_j$ commute), the solution of the equation $(A\,h_0)\,(x) = 0$ is given by $h_0(x) = (\pi\cstFock)^{-n/4}\,e^{-x^2/2\cstFock}$, where the constant has been chosen to normalize $h_0$ on $L^2(\R^n,dx)$. \\
      Moreover $(W_\cstFock\, h_0) (z) = (\pi\cstFock)^{-n/4}\int_{x\in \R^n} e^{(-z^2 + 2\sqrt{2}\,xz - x^2)/2\cstFock} \,  (\pi\cstFock)^{-n/4}\, e^{-x^2/2\cstFock}   \, dx = 1$.
      Intertwining \eqref{WA_zW} with $W_\cstFock\, h_0=1$, we get for $\alpha\in \N^n$, $W_\cstFock \big((A^*)^{\alpha}\, h_0\big) (z)= z^{\alpha}$, where $(A^*)^{\alpha} = \prod_{j=1}^n (A_j^*)^{\alpha_j}$. The functions
      \begin{align}
        \label{halpha}
        h_\alpha(x)\vc (\cstFock^{\abs{\alpha}}\,\alpha!)^{-1/2}\,\big((A^*)^{\alpha}\, h_0\big) (x), \quad x\in \R^n
      \end{align}
      are the Hermite functions and they form an orthonormal basis of $L^2(\R^n,dx)$. \\
      So $W_\cstFock$ maps unitarily the space $L^2(\R^n)$ onto $\mathscr{F}_\cstFock$ since
      $W_\cstFock : h_\alpha \mapsto (\cstFock^{\abs{\alpha}}\alpha!)^{-1/2} \, z^{\alpha} = u_\alpha$.

      Since $W_\cstFock$ is unitary, its inverse $W_\cstFock^{-1}$ is equal to its Hilbert adjoint and for $g \in L^2(\R^n,dx)$,
      \begin{align*}
        \scalp{W_\cstFock \,g}{f}_{\mathscr{F}_\cstFock} = & \, (\pi\cstFock)^{-n} \int_{z\in\C^n} (\pi\cstFock)^{-n/4}\int_{x\in\R^n} e^{(-z^2+2\sqrt{2}\,xz - x^2)/2\cstFock} \, g(x)\,\overline{f(z)} \, e^{-\abs{z}^2/\cstFock} \, dx\,d\mu(z) \\
        = & \, \int_{x\in\R^n} g(x)\, \Big(\overline{  (\pi\cstFock)^{-5n/4} \int_{z\in\C^n}  e^{(-\bar{z}^2 + 2\sqrt{2}\,x\bar{z} - x^2)/2\cstFock} \, f(z) \, e^{-\abs{z}^2/\cstFock} \,d\mu(z)}\Big)\,dx \\
        = & \, \scalp{g}{W_\cstFock^* f}_{L^2(\R^n,dx)}.
      \end{align*}
    \end{proof}

  \noindent \\
  \noindent {\bf Acknowledgements}\\
  The authors would like to thank Louis Boutet de Monvel for helpful discussions about the theory of generalized Toeplitz operators.

%%%%%%%%%%%%%%%%%%%%%%%%%%%%%%%%%%%%%%%%%%%%
% Bibliography %
%%%%%%%%%%%%%%%%%%%%%%%%%%%%%%%%%%%%%%%%%%%%

\end{document}